\documentclass[10pt, a4paper]{article} 
\usepackage[square,numbers,sort&compress]{natbib}
\usepackage[shortlabels]{enumitem}
\usepackage{setspace}
\usepackage{subfloat}
\usepackage{mathtools}
\usepackage{amsmath,amsfonts,amssymb,amsthm,mathrsfs, bm}
\usepackage[utf8]{inputenc}
\usepackage[swedish, english]{babel}
\usepackage{csquotes}
\usepackage{subfig}
\usepackage{graphicx}
\usepackage{supertabular}                                             
\usepackage[table]{xcolor} 
\usepackage{textgreek} 
\usepackage{graphicx}
\usepackage{epstopdf}
\usepackage{microtype}
\usepackage{multirow}
\usepackage{hyperref}
\usepackage{xcolor}
\usepackage{authblk}
\usepackage{datetime2}	
\usepackage[margin=1in]{geometry} % Change of margins
\usepackage{setspace}	% for double spacing
\theoremstyle{plain}% Theorem-like structures provided by amsthm.sty
\newtheorem{theorem}{Theorem}[section]
\newtheorem{lemma}{Lemma}[section]
\newtheorem{corollary}{Corollary}[section]

\usepackage{amsmath,amsfonts,amssymb,amsthm,mathrsfs, bm}
\usepackage{stmaryrd}

\usepackage{graphicx}
\usepackage{tikz}
\usepackage[utf8]{inputenc}
\usepackage{pgfplots} 
\usepackage{pgfgantt}
\usepackage{pdflscape}
\pgfplotsset{compat=newest} 
\pgfplotsset{plot coordinates/math parser=false}

\usepackage{endnotes}
\usepackage{commath}
\usepackage{gensymb}
\usepackage{mathtools}
\usepackage{tikz} 
\usetikzlibrary{calc,quotes,angles}
\usepackage{tkz-euclide}
\usetikzlibrary{automata,positioning} 

\usepackage{xcolor}
\usepackage{tikz,tkz-euclide,siunitx}
\usepackage{etoolbox} % for \patchcmd

\usetikzlibrary{quotes,arrows.meta,angles,calc,shadings,positioning,babel}

\newtheorem{remark}{Remark}[section]

\usepackage{endnotes}
\usepackage{commath}
\usepackage{gensymb}
\usepackage{tcolorbox}
\makeatletter
\patchcmd{\tkz@DrawLine}{\begingroup}{\begingroup\makeatletter}{}{}
\makeatother

\allowdisplaybreaks

\DeclareMathOperator{\atan}{arctan}

\DeclareMathOperator{\esup}{ess\, sup}

\makeatletter
\newcommand\makebig[2]{%
  \@xp\newcommand\@xp*\csname#1\endcsname{\bBigg@{#2}}%
  \@xp\newcommand\@xp*\csname#1l\endcsname{\@xp\mathopen\csname#1\endcsname}%
  \@xp\newcommand\@xp*\csname#1r\endcsname{\@xp\mathclose\csname#1\endcsname}%
}
\makeatother

\providecommand*{\ped}[1]{%
\ensuremath{_\textnormal{#1}}}

\providecommand*{\eu}%
{\ensuremath{\mathrm{e}}}
\providecommand*{\im}%
{\ensuremath{\mathrm{i}}}
\providecommand*{\GammaF}%
{\ensuremath{\mathrm{\Gamma}}}
\providecommand*{\BetaF}%
{\ensuremath{\mathrm{\Beta}}}

\makebig{biggg} {3.0}
\makebig{Biggg} {3.5}
\makebig{bigggg}{4.0}
\makebig{Bigggg}{4.5}
\makebig{biggggg}{5.0}
\makebig{Biggggg}{5.5}
\usepackage{longtable}
\usepackage{xcolor}
\DeclareMathSymbol{\Gamma}{\mathalpha}{letters}{"00}
\DeclareMathSymbol{\Delta}{\mathalpha}{letters}{"01}
\DeclareMathSymbol{\Theta}{\mathalpha}{letters}{"02}
\DeclareMathSymbol{\Lambda}{\mathalpha}{letters}{"03}
\DeclareMathSymbol{\Xi}{\mathalpha}{letters}{"04}
\DeclareMathSymbol{\Pi}{\mathalpha}{letters}{"05}
\DeclareMathSymbol{\Sigma}{\mathalpha}{letters}{"06}
\DeclareMathSymbol{\Upsilon}{\mathalpha}{letters}{"07}
\DeclareMathSymbol{\Phi}{\mathalpha}{letters}{"08}
\DeclareMathSymbol{\Psi}{\mathalpha}{letters}{"09}
\DeclareMathSymbol{\Omega}{\mathalpha}{letters}{"0A}

\definecolor{matblue}{rgb}{0.0000,0.4470,0.7410}
\definecolor{matred}{rgb}{0.8500,0.3250,0.0980}
\definecolor{matyellow}{rgb}{0.9290,0.6940,0.1250}
\definecolor{matpurple}{rgb}{0.4940,0.1840,0.5560}
\definecolor{matgreen}{rgb}{0.4660,0.6740,0.1880}
\definecolor{matcyan}{rgb}{0.3010,0.7450,0.9330}
\definecolor{matmaroon}{rgb}{0.6350,0.0780,0.1840}

\usepackage{hyperref}
\hypersetup{
    colorlinks = true,
    linkbordercolor = {white},
            linkcolor=blue,
            bookmarks=true,
            filecolor=blue,
            urlcolor=blue,
            citecolor=blue
}

\newtcolorbox[auto counter]{modelbox}[2][]{%
  colback=white, colframe=black,
  fonttitle=\bfseries,
  title=Model~\thetcbcounter: #2,
  label=#1
}

\begin{document}

\title{Bare and stretched string tyre models\\ with distributed FrBD dynamics}
\date{}
\author[a,b]{Luigi Romano\thanks{Corresponding author. Email: luigi.romano@liu.se.}}
%\author[c]{Igo Besselink}
%\author[a]{Jan Aslund}
%\author[a]{Erik Frisk}
\affil[a]{\footnotesize{Division of Vehicular Systems, Department of Electrical Engineering, Linköping University, SE-581 83 Linköping, Sweden}}
%\affil[b]{\footnotesize{Department of Engineering Cybernetics, Norwegian University of Science and Technology, O. S. Bragstads plass 2, NO-7034, Trondheim, Norway}}
%\affil[b]{\footnotesize{Department of Mechanical Engineering, Eindhoven University of Technology, Groene Loper 1, 5612 AZ Eindhoven, the Netherlands}}
\affil[b]{\footnotesize{Control Systems Technology Group, Department of Mechanical Engineering, Eindhoven University of Technology, Groene Loper 1, 5612 AZ Eindhoven, the Netherlands}}
%\affil[c]{\footnotesize{Dynamics and Control Group, Department of Mechanical Engineering, Eindhoven University of Technology, Groene Loper 1, 5612 AZ Eindhoven, the Netherlands}}

\maketitle

\begin{abstract}
This paper presents a novel class of string tyre models with FrBD friction dynamics. By modelling the distributed carcass and tread deformations with string-like equations, the resulting formulation leads to a system of semilinear parabolic \emph{partial differential equations} (PDEs) that describe the evolution of the tyre states without explicitly distinguishing between stick and slip regimes. Rigorous stability and passivity analyses are also conducted using a Lyapunov-based approach, establishing boundedness of the distributed states and energy dissipation during rolling contact. The proposed Lyapunov function admits a clear physical interpretation as the total elastic energy stored in the tyre, enabling a direct link between mechanical energy storage and frictional dissipation due to slip losses. The steady-state and transient behaviours of the model are investigated both numerically and experimentally, revealing that the new formulation can satisfactorily reproduce nonlinear relaxation phenomena excited by step slip inputs.
The resulting model provides a physically interpretable, mathematically well-posed, and computationally efficient basis for advanced vehicle dynamics simulations and control-oriented applications.
\end{abstract}
\section*{Keywords}
Tyre dynamics; string tyre models; rolling contact; friction modelling; distributed parameter systems; semilinear systems

%In their original form, however, such models are essentially quasi-static and cannot capture relaxation phenomena, memory effects, or the transient response to rapidly varying slip and spin inputs.

\section{Introduction}\label{intro}
The process of tyre forces and moment generation govern essential vehicle dynamics such as handling, braking, traction, and stability, and its accurate representation is critical for both simulation-based design and advanced control applications \cite{Rajamani,Nielsen,Savaresi,LuGreControl2,deCastro2,Poussot,Wang,Solyom,LuGreControl3,deCastro,Fors1,Fors2,Fors3,Fors4}. Whereas steady-state tyre characteristcs are relatively well understood and may be reproduced with high fidelity using empirical and semi-physical models, a reliable description of transient tyre behaviour remains a persistent challenge \cite{Pacejka2,Guiggiani}.

Early analytical models, including brush \cite{Pacejka2,Guiggiani,LibroMio} and string-like formulations \cite{Schlippe,Segel,Segel1,Segel2,Higuchi1,Higuchi2,Pauwelussen}, have provided a physically interpretable framework for relating the local deformation of the tyre tread and carcass to the resulting forces and moments. These models, which are rooted in rolling contact mechanics and classical friction laws \cite{Johnson}, have played a central role in the development of modern tyre theory and continue to underpin widely used semi-empirical formulations \cite{Rill1,PAC}. In these descriptions, internal state variables typically represent bristle deflections or distributed contact deformations, leading to systems of hyperbolic PDEs describing the transport of shear stresses along the contact patch \cite{USB,Meccanica2}. Combined with local Coulomb-Amontons friction laws, similar formulations naturally predict nonlinear relaxation phenomena, and can reproduce experimentally observed delays in force build-up. Nevertheless, their hyperbolic structure imposes intrinsic constraints: boundary conditions cannot be imposed simultaneously at both ends of the contact patch, spatial discontinuities may arise at the leading or trailing edge  \cite{Pacejka2,Higuchi1,Higuchi2}, and oscillatory responses are often observed under time-varying slip excitation. Additionally, Coulomb-like friction models can predict the existence of multiple stick and slip zones within the tyre's contact patch \cite{Pacejka2,Guiggiani,LibroMio,Higuchi1,Higuchi2,USB,Meccanica2}, posing enormous difficulties to the rigorous mathematical analyses of these descriptions. Such features also complicate both numerical implementation, particularly when these models are embedded within vehicle-level simulations.

More recently, the development of dynamic friction models of LuGre \cite{Sorine,TsiotrasConf,Tsiotras1,Tsiotras3,Deur0,Deur1,Deur2,Gauterin,Gauterin2} and FrBD type \cite{FrBD,2DFrBD,2DFrBD3} has advanced the modelling of transient tyre-road interaction at the local level. These formulations introduce internal friction states that capture pre-sliding displacement, hysteresis, and velocity-dependent friction effects, and have been successfully applied to tyre dynamics under combined slip and spin conditions. When coupled with distributed contact descriptions, FrBD-based tyre models provide a powerful framework for describing transient rolling contact phenomena without explicitly differentiating between stick and slip behaviours, thus overcoming the intrinsic disadvantages Coulomb friction. However, existing formulations typically retain a transport-dominated structure \cite{FrBD,2DFrBD,2DFrBD3}, and the influence of tyre structural compliance on the mathematical nature of the governing equations has received limited attention.

In contrast, experimental and numerical evidence indicates that tyre relaxation phenomena cannot be adequately described by pure transport mechanisms. In particular, increasing slip levels are often associated with progressively shorter transients and a pronounced attenuation of high-frequency slip excitations \cite{Pacejka2,Higuchi1,Higuchi2,Besselink}. These observations suggest the presence of diffusive effects within the tyre-road interaction process, arising from the distributed compliance of the tread and carcass, which are not captured by classical hyperbolic formulations without Coulomb-like friction.
Motivated by these considerations, this paper introduces the simplest string-based tyre model with distributed FrBD friction, referred to as the \emph{friction with string dynamics} (FrSD) model. The key idea is to combine a nonlocal constitutive description of tyre structural deformation with a dynamic friction law of FrBD type. Unlike traditional brush models, which assume a simple elastic foundation to describe the deflection of bristle elements \cite{Pacejka2,Guiggiani,LibroMio}, the constitutive relationships adopted in this work explicitly relate internal stress resultants to spatial derivatives of the string deformation field through a simplified one-dimensional constitutive representation, leading to governing equations that are parabolic for all nonzero slip velocities. This represents a fundamental departure from classical transport-based rolling contact models and has important physical, mathematical, and numerical implications.

Indeed, from a physical standpoint, the parabolic structure naturally introduces diffusive mechanisms that account for relaxation phenomena modulated by the magnitude of the slip and spin inputs. This enables the model to reproduce experimentally observed trends such as faster force build-up at higher slip levels and the attenuation of oscillations under high-frequency excitation. From a modelling perspective, the parabolic nature of the equations allows boundary conditions to be imposed simultaneously at both ends of the contact patch, resulting in smooth deformation profiles without artificial kinks at the contact edges. From a systems-theoretic viewpoint, the resulting formulation is well suited for stability analysis and for integration into vehicle-level models, where robustness and numerical reliability are essential.

In this picture, it is essential to clarify that several alternative modelling approaches can reproduce nonlinear transient tyre behaviour and attenuate rapidly varying slip inputs. For instance, at the lowest level of complexity, nonlinear relaxation-length formulations extend steady-state tyre characteristics with one or more lumped relaxation equations. Such models are computationally efficient, numerically robust, and idoneous for simulation and control-oriented applications, but their relaxation dynamics is often postulated phenomenologically rather than being derived from distributed contact mechanics \cite{Pacejka2,LibroMio}. In contrast, at a higher level of structural fidelity, beam models and flexible-ring models account for belt bending, structural inertia, and vibration modes, and are therefore more appropriate when resonant, wave-propagation, ride, or NVH phenomena are of interest \cite{Soedel,Zegelaar}. Their additional states and parameters, however, generally increase the computational and identification effort. The FrSD formulation proposed in this paper aims instead to provide an intermediate description in which slip-dependent relaxation and high-frequency attenuation behaviours emerge from the coupling between distributed friction dynamics and a minimal structural model, whilst retaining a semilinear PDE representation that is computationally manageable and amenable to rigorous well-posedness, stability, and passivity analyses.

In this context, the main contributions of this work can be summarised as follows:
\begin{enumerate}[(i)]
\item The derivation of a novel string-based tyre model with distributed FrBD friction dynamics, grounded in a nonlocal constitutive description of tyre deformation,

\item The identification and analysis of the parabolic character of the resulting governing equations, highlighting its consequences for transient tyre behaviour, including mathematical properties such as stability and passivity,

%\item The formulation of consistent expressions for tyre forces and aligning moments under combined slip and spin conditions,

\item Numerical and experimental validations of both steady-state and transient tyre responses, demonstrating the model's ability to capture relaxation effects, slip-dependent dynamics, and low-pass filtering behaviour.
\end{enumerate}

The paper is organised as follows. Section~\ref{sect:model} introduces the constitutive equations describing tyre structural compliance and the adopted dynamic friction law, derives the governing equations of the FrSD tyre model, and presents expressions for the resulting forces and moments. The mathematical properties of the new string-like formulation are then discussed in Sect.~\ref{sect:math}, where stability and passivity are studied rigorously. Section~\ref{sect:NumAnDecp} illustrates the steady-state and transient behaviour of the model through numerical simulations and experimental validation. Finally, concluding remarks and perspectives for future work are provided in Sect. ~\ref{sect:conclusions}.

\subsection*{Notation}
In this paper, $\mathbb{R}$ denotes the set of real numbers; $\mathbb{R}_{>0}$ and $\mathbb{R}_{\geq 0}$ indicate the set of positive real numbers excluding and including zero, respectively. The set of positive integer numbers is indicated with $\mathbb{N}$, whereas $\mathbb{N}_{0}$ denotes the extended set of positive integers including zero, i.e., $\mathbb{N}_{0} = \mathbb{N} \cup \{0\}$.
The set of $n\times m$ matrices with values in $\mathbb{F}$ ($\mathbb{F} = \mathbb{R}$, $\mathbb{R}_{>0}$, $\mathbb{R}_{\geq0}$) is denoted by $\mathbf{M}_{n\times m}(\mathbb{F})$ (abbreviated as $\mathbf{M}_{n}(\mathbb{F})$ whenever $m=n$). $\mathbf{Sym}_n(\mathbb{R})$ represents the group of symmetric matrices with values in $\mathbb{R}$; the identity matrix on $\mathbb{R}^n$ is indicated with $\mathbf{I}_n$. A positive-definite matrix is noted as $\mathbf{M}_n(\mathbb{R}) \ni \mathbf{Q} \succ \mathbf{0}$; a positive semidefinite one as $\mathbf{M}_n(\mathbb{R}) \ni \mathbf{Q} \succeq \mathbf{0}$. 
The standard Euclidean norm on $\mathbb{R}^n$ is indicated with $\norm{\cdot}_2$; operator norms are simply denoted by $\norm{\cdot}$.
Given a domain $\Omega$ with closure $\overline{\Omega}$, $L^p(\Omega;\mathcal{Z})$ and $C^k(\overline{\Omega};\mathcal{Z})$ ($p, k \in \{1, 2, \dots, \infty\}$) denote respectively the spaces of $L^p$-integrable functions and $k$-times continuously differentiable functions on $\overline{\Omega}$ with values in $\mathcal{Z}$ (for $T = \infty$, the interval $[0,T]$ is identified with $\mathbb{R}_{\geq 0}$). For $\alpha \in (0,1)$, $C^\alpha(\overline{\Omega};\mathcal{Z})$ denotes the Banach space of Holder-continuous functions with values in $\mathcal{Z}$. $L^2(\Omega;\mathbb{R}^n)$ denotes the Hilbert space of square-integrable functions on $\Omega$ with values in $\mathbb{R}^n$, endowed with inner product $\langle \bm{\zeta}_1, \bm{\zeta}_2 \rangle_{L^2(\Omega;\mathbb{R}^n)} = \int_\Omega \bm{\zeta}_1^{\mathrm{T}}(x)\bm{\zeta}_2(x) \dif x$ and induced norm $\norm{\bm{\zeta}(\cdot)}_{L^2(\Omega;\mathbb{R}^n)}$. The Hilbert space $H^1(\Omega;\mathbb{R}^n)$ consists of functions $\bm{\zeta}\in L^2(\Omega;\mathbb{R}^n)$ whose weak derivative also belongs to $L^2(\Omega;\mathbb{R}^n)$. For a function $f :\Omega \mapsto \mathbb{R}$, the sup norm is defined as $\norm{f(\cdot)}_\infty \triangleq \esup_{\Omega} \abs{f(\cdot)}$; $f : \Omega \mapsto \mathbb{R}$ belongs to the space $L^\infty(\Omega;\mathbb{R})$ if $\norm{f(\cdot)}_\infty < \infty$. A function $f \in C^0(\mathbb{R}_{\geq 0}; \mathbb{R}_{\geq 0})$ belongs to the space $\mathcal{K}$ if it is strictly increasing and $f(0) = 0$; $f \in \mathcal{K}$ belongs to the space $\mathcal{K}_\infty$ if it is unbounded. Finally, a function $f \in C^0(\mathbb{R}_{\geq 0}^2; \mathbb{R}_{\geq 0})$ belongs to the space $\mathcal{KL}$ if $f(\cdot,t) \in \mathcal{K}$ and is strictly decreasing in its second argument, with $\lim_{t\to \infty}f(\cdot,t) = 0$.

\section{Model derivation}\label{sect:model}
The present section is devoted to the derivation of the new FrSD tyre model. The main equations are first introduced in Sect.~\eqref{sect:eqs}, and then combined together in Sect.~\ref{sect:FrBD} via an application of the \emph{Implicit Function Theorem}.

Before moving to the model development, however, it is beneficial to first state the main assumptions underlying the proposed formulation, along with their implications on its predictive capabilities:
\begin{enumerate}[(i)]
\item The tyre structure is represented by longitudinal and lateral strings, which leads to neglecting belt bending, cross-sectional deformation, and higher-order structural modes that can be described by beam- or flexible-ring-based models,

\item Structural inertia is neglected; consequently, the model is not intended to reproduce resonant vibrations, wave propagation, or inertia-induced oscillations,

\item The contact-patch length and normal-pressure distribution are supposed to be prescribed. Therefore, the proposed formulation does not capture load-dependent footprint changes, road compliance, or normal contact dynamics,

\item The distributed contact force is assumed to vanish outside the contact region. Together with the adopted string-like constitutive laws, this implies that the deformation decays exponentially outside the contact patch. The resulting continuation provides a smooth idealisation of the string dynamics beyond the contact region, but does not aim to reproduce accurately the deformation of the complete tyre away from the contact patch,

\item The considered kinematics account for longitudinal and lateral slip and moderate spin, whereas the effect of large camber angles, turn-slip speeds, and changes in wheel orientation across the contact patch are excluded.
\end{enumerate}
The assumptions above are adopted deliberately to obtain a minimal distributed model that remains physically interpretable and amenable to rigorous well-posedness, stability, and passivity analyses.

\subsection{Constitutive equations, kinematic constraints, and friction model}\label{sect:eqs}
The key ingredients for deriving the new string-like description are the constitutive equations, which relate the tyre deformations to the generated forces and moments, the kinematic constraints, linking the sliding and rigid relative velocities to the treadband and carcass deflections, and the friction model. These are described in detail in Sections~\ref{sect:constitutibe} and~\ref{sect:friction}, respectively.

\begin{figure}
\centering
\includegraphics[width=1\linewidth]{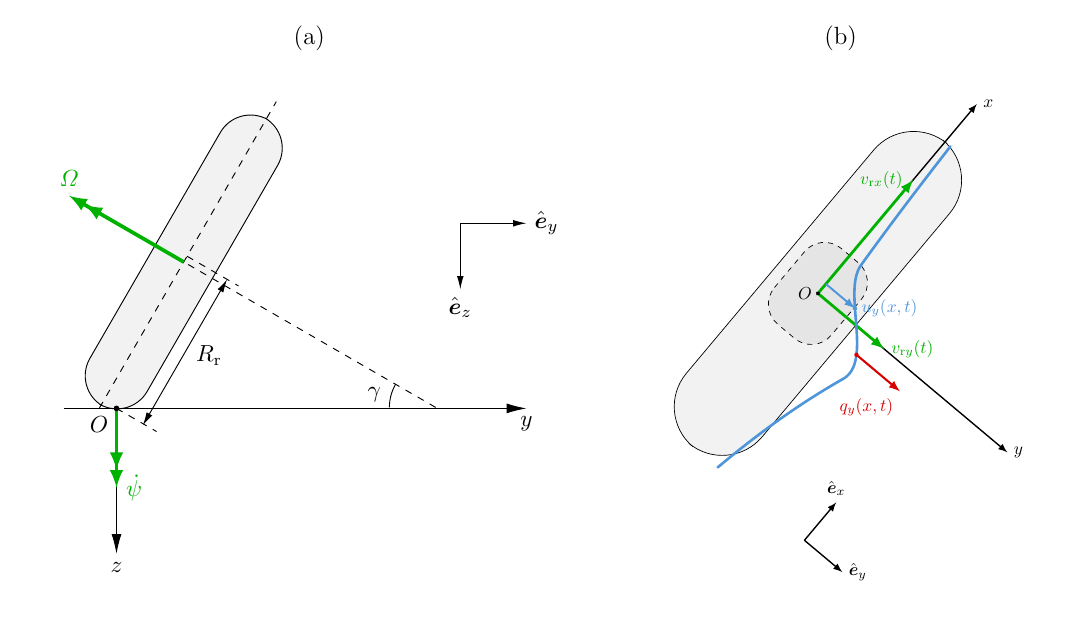} 
\caption{Tyre kinematics and contact patch reference frame: (a) schematic of the tyre-wheel systems with angular velocities; (b) schematic of the lateral tyre-road interaction in the $x-y$ plane in the absence of longitudinal slip and spin ($\varphi = \gamma = 0$), with the deformation profile of the string in the lateral direction. Line styles: kinematic quantities (green lines), deformations (blue lines), red stresses and forces (red lines).}
\label{fig:tireSchematic}
\end{figure} 

\subsubsection{Constitutive differential equations and kinematic constraints}\label{sect:constitutibe} 
As illustrated in Fig.~\ref{fig:tireSchematic}, the tyre-road contact problem is studied in this paper in a right-handed reference frame $(O;x,y,z)$ whose axes are oriented as follows: the $x$ axis is aligned with the principal direction of motion, the $z$ axis points downward towards the road surface, assumed to be nominally flat and identified with the plane $\Pi = \{(x,y,z) \in \mathbb{R}^3 \mathrel{|} z = 0\}$, and the $y$ axis is defined such that the resulting coordinate system is right-handed. The origin $O$ is located at the centre of the tyre's contact patch $\mathscr{P}$, which, under line contact assumptions, is given by $\mathscr{P} = \{ x \in \mathbb{R} \mathrel{|} -a\leq x \leq a\}$, with $a \in \mathbb{R}_{>0}$ denoting its semilength. Accordingly, the longitudinal and lateral displacements of the treadband and carcass centrelines may be described by string-like equations of type \cite{Higuchi1,Higuchi2}: 
\begin{align}\label{eq:ODE_SC}
\mathbf{S}\dpd[2]{\bm{u}(x,t)}{x}-\mathbf{K}\bm{u}(x,t) = -\bm{q}(x,t), \quad (x,t) \in (-a,a)\times (0,T), 
\end{align}
where $\mathbb{R}^2\ni \bm{u}(x,t) = [u_x(x,t)\; u_y(x,t)]^{\mathrm{T}}$ collects the deflections of the tyre treadband and carcass, respectively, $\mathbb{R}^2 \ni \bm{q}(x,t) = [q_x(x,t)\; q_y(x,t)]^{\mathrm{T}}$ denotes the force per unit of length acting inside the contact patch, and the matrices $\mathbf{Sym}_2(\mathbb{R}) \ni \mathbf{S}\succ \mathbf{0}$ and $\mathbf{Sym}_2(\mathbb{R}) \ni \mathbf{K} \succ \mathbf{0}$ are given by
\begin{subequations}
\begin{align}
\mathbf{S} & = \begin{bmatrix} EA & 0 \\ 0 & S\end{bmatrix}, \\
\mathbf{K} & = \begin{bmatrix} k_x & 0 \\ 0 & k_y\end{bmatrix},
\end{align}
\end{subequations}
being $S \in \mathbb{R}_{>0}$ the effective tension of the string, $E \in \mathbb{R}_{>0}$ the Young's modulus, $A \in \mathbb{R}_{>0}$ the cross-sectional area, and $k_x, k_y \in \mathbb{R}_{>0}$ longitudinal and lateral stiffnesses per unit of length \cite{Higuchi1,Higuchi2}. The longitudinal and lateral equations obtained from Eq.~\eqref{eq:ODE_SC} are commonly referred to as \emph{bare} and \emph{stretched string tyre models}.

The sliding velocity of the tyre treadband and carcass particles, $\bm{v}\ped{s}(x,t) \in \mathbb{R}^2$, is given by \cite{Guiggiani,LibroMio}
\begin{align}\label{eq:slidingvel}
\begin{split}
\bm{v}\ped{s}(x,t) & = \bm{v}\ped{r}(x,t) + \dot{\bm{u}}(x,t) = \bm{v}\ped{r}(x,t) + \dod{\bm{u}(x,t)}{t} \\
& = \bm{v}\ped{r}(x,t) + \dpd{\bm{u}(x,t)}{t} - V\ped{r}(t)\dpd{\bm{u}(x,t)}{x}, \quad (x,t) \in (-a,a)\times(0,T),
\end{split}
\end{align}
where $\mathbb{R}^2 \ni \bm{v}\ped{r}(x,t) \triangleq [v_{\mathrm{r}x}(x,t) \; v_{\mathrm{r}y}(x,t)]^{\mathrm{T}}$ denotes the rigid relative velocity, and $V\ped{r}(t) \in [V\ped{min}, V\ped{max}]$, with $0 < V\ped{min} \leq V\ped{max}$, is the \emph{rolling speed}, reading $V\ped{r}(t) = \Omega(t)R\ped{r}(t)$, where $\Omega(t) \in \mathbb{R}_{>0}$ represents the angular velocity of the tyre around its main axis, and $R\ped{r}(t) \in \mathbb{R}_{>0}$ the \emph{effective rolling radius} \cite{Guiggiani}. 

Together, Eqs.~\eqref{eq:ODE_SC} and~\eqref{eq:slidingvel} constitute the first ingredient for the development of the new tyre model.
More specifically, Eq.~\eqref{eq:ODE_SC} postulates a constitutive relationship linking the deformation $\bm{u}(x,t)$ and the force (or stress) $\bm{q}(x,t)$, whereas~\eqref{eq:slidingvel} provides the kinematic constraints relating the tyre deflections with the rigid and sliding velocities. To further proceed with the derivation, the following matrices are conveniently introduced:
\begin{subequations}
\begin{align}
\mathbf{\Sigma}_0(x) & = \begin{bmatrix} \sigma_{0x}(x) & 0 \\0 & \sigma_{0y}(x) \end{bmatrix} \triangleq \dfrac{\mathbf{K}}{p(x)}, \\
\mathbf{\Sigma}_1(x) &= \begin{bmatrix} \sigma_{1x}(x) & 0 \\0 & \sigma_{1y}(x) \end{bmatrix} \triangleq \dfrac{\mathbf{S}}{p(x)},
\end{align}
\end{subequations}
where $p \in C^1([-a,a];\mathbb{R}_{>0})$ denotes the vertical pressure distribution inside the tyre's contact patch. Accordingly, the original ODE~\eqref{eq:ODE_SC} may be recast as 
\begin{align}\label{eq:ODE_Sigma}
\mathbf{\Sigma}_1(x)\dpd[2]{\bm{u}(x,t)}{x}-\mathbf{\Sigma}_0(x)\bm{u}(x,t) = -\bm{f}(x,t), \quad (x,t) \in (-a,a)\times (0,T),
\end{align}
where $\mathbb{R}^2 \ni \bm{f}(x,t) = [f_x(x,t)\; f_y(x,t)]^{\mathrm{T}} \triangleq \bm{q}(x,t)/p(x)$ represents the tangential force per unit of vertical load.
The \emph{boundary conditions} (BCs) for both Eqs.~\eqref{eq:ODE_SC} and~\eqref{eq:ODE_Sigma} may be formulated by observing that, outside the contact patch, that is, for $x \in (-\infty,-a]$ and $x \in [a,\infty)$, the distributed contact force vanishes, i.e., $\bm{q}(x,t)=\bm{0}$. Under the constitutive relationship~\eqref{eq:ODE_SC}, the corresponding homogeneous differential equation admits exponentially decaying solutions. Enforcing continuity with these outer solutions leads to the Robin, or mixed, BCs \cite{Pacejka2,Higuchi1,Higuchi2}: 
\begin{subequations}\label{eq:Bcs}
\begin{align}
&\mathbf{\Lambda}\dpd{\bm{u}(a,t)}{x} + \bm{u}(a,t) = \bm{0}, \\
& \mathbf{\Lambda}\dpd{\bm{u}(-a,t)}{x} - \bm{u}(-a,t) = \bm{0}, \quad t\in(0,T),
\end{align}
\end{subequations}
where $\mathbf{Sym}_2(\mathbb{R}) \cap \mathbf{GL}_2(\mathbb{R}) \ni \mathbf{\Lambda} \succ \mathbf{0}$ collects the \emph{longitudinal} and \emph{lateral relaxation lengths} $\lambda_x, \lambda_y \in \mathbb{R}_{>0}$:
\begin{align}
\mathbf{\Lambda} = \begin{bmatrix} \lambda_x & 0 \\ 0 & \lambda_y\end{bmatrix} \triangleq \begin{bmatrix} \sqrt{\dfrac{EA}{k_x}} & 0 \\ 0 & \sqrt{\dfrac{S}{k_y}}\end{bmatrix}.
\end{align}
%It is worth noting that the matrix $\mathbf{\Lambda}$ is supposed to be constant over space. This assumption is essential for Eq.~\eqref{eq:Bcs} to hold.

\subsubsection{Friction model}\label{sect:friction}
In addition to the constitutive relationship~\eqref{eq:ODE_Sigma} and kinematic constraints~\eqref{eq:slidingvel} for the tyre structural elements, it is necessary to specify a local model for the nondimensional friction force $\bm{f}\ped{r} \in \mathbb{R}^2$ as a function of the sliding velocity. Based on \cite{Sorine,Tsiotras3}, the following relationship is considered in this paper:
\begin{align}\label{eq:frModified}
\bm{f}\ped{r}(\bm{v}\ped{s}) = -\dfrac{\mathbf{M}^2(\bm{v}\ped{s})\bm{v}\ped{s}}{\norm{\mathbf{M}(\bm{v}\ped{s})\bm{v}\ped{s}}_{2,\varepsilon}},
\end{align}
where $\varepsilon \in \mathbb{R}_{\geq 0}$ represents a regularisation parameter, $\norm{\cdot}_{2,\varepsilon} \in C^0(\mathbb{R}^2;\mathbb{R}_{\geq 0})$ is a regularisation of the Euclidean norm $\norm{\cdot}_2$ for $\varepsilon\in \mathbb{R}_{>0}$, often converging uniformly to $\norm{\cdot}_2$ in $C^0(\mathbb{R}^2;\mathbb{R}_{\geq 0})$ for $\varepsilon \to 0$ (e.g., $\norm{\bm{y}}_{2,\varepsilon }= \sqrt{\norm{\bm{y}}_2^2 +\varepsilon}$), and with $\norm{\cdot}_{2,\varepsilon} \in C^1(\mathbb{R}^2;\mathbb{R}_{\geq 0})$ for $\varepsilon \in \mathbb{R}_{>0}$, and
\begin{align}\label{eq;matrixM}
\mathbf{M}(\bm{v}\ped{s}) = \begin{bmatrix} \mu_{x}(\bm{v}\ped{s}) &  0 \\ 0  & \mu_{y}(\bm{v}\ped{s})\end{bmatrix}
\end{align}
is a symmetric, positive definite matrix of friction coefficients, i.e., $\mathbf{Sym}_2(\mathbb{R}) \ni \mathbf{M}(\bm{y}) \succ \mathbf{0}$ for all $\bm{y} \in \mathbb{R}^2$. In this manuscript, it is generally assumed that $\mathbf{M} \in C^0(\mathbb{R}^2;\mathbf{Sym}_2(\mathbb{R}))$ and $\varepsilon \in \mathbb{R}_{>0}$. Indeed, the adoption of a regularisation parameter $\varepsilon \in \mathbb{R}_{>0}$ prevents singular behaviour when the sliding velocity approaches zero. Without regularisation, the friction law~\eqref{eq:frModified} becomes nondifferentiable at $\bm{v}\ped{s} = \bm{0}$, which would introduce singular coefficients in the governing PDEs of the system and compromise the theoretical results established in Sect.~\ref{sect:math}. The parameter therefore serves both mathematical and numerical purposes and is not merely introduced for computational convenience. 

For many practical applications, and especially in isotropic conditions, $\mathbf{M}(\bm{v}\ped{s})$ may be specified as $\mathbf{M}(\bm{v}\ped{s}) = \mu(\bm{v}\ped{s})\mathbf{I}_2$, where $\mu : \mathbb{R}^2 \mapsto [\mu\ped{min},\infty)$, with $ \mu\ped{min} \in \mathbb{R}_{>0}$, provides an analytical expression for the friction coefficient. A very common expression for $\mu(\bm{v}\ped{s})$ is, for instance,
\begin{align}\label{eq:muExample}
\mu(\bm{v}\ped{s}) = \mu\ped{d} + (\mu\ped{s}-\mu\ped{d})\exp\Biggl(-\biggl(\dfrac{\norm{\bm{v}\ped{s}}_2}{v\ped{S}}\biggr)^{\delta\ped{S}}\Biggr)+ \mu\ped{v}(\bm{v}\ped{s}),
\end{align}
where $\mu\ped{s},\mu\ped{d} \in \mathbb{R}_{>0}$ denote the static and dynamic friction coefficients, $v\ped{S} \in \mathbb{R}_{>0}$ indicates the Stribeck velocity, $\delta\ped{S} \in \mathbb{R}_{\geq 0}$ the Stribeck exponent, and $\mu\ped{v} : \mathbb{R}^2 \mapsto \mathbb{R}_{\geq 0}$ captures the viscous friction.

Combining the constitutive equations~\eqref{eq:ODE_SC} with the kinematic constraints~\eqref{eq:slidingvel} and the friction model~\eqref{eq:frModified}, it is possible to derive the new string-like description of the tyre dynamics. This is addressed next in Sect.~\ref{sect:FrBD}.

\subsection{String tyre model with FrBD dynamics (FrSD)}\label{sect:FrBD}
%This section introduces and analyses the new string-like tyre model. More specifically, its governing equations are presented in Sect.~\ref{sect.Eaus}, whereas Sect.~\ref{sect:numerical} discusses its qualitative dynamic behaviour.

%\subsection{Model equations}\label{sect.Eaus}
Starting with Eqs.~\eqref{eq:ODE_SC},~\eqref{eq:slidingvel}, and~\eqref{eq:frModified}, Theorem~\ref{thm:Theorem1} will guide the derivation of the governing equations of the FrSD tyre model.

\begin{theorem}[Edwards \cite{Edwards}]\label{thm:Theorem1}
Suppose that the mapping $\bm{H} : \mathbb{R}^{m+n}\mapsto \mathbb{R}^n$ is $C^1$ in a neighbourhood of a point $(\bm{x}^\star,\bm{y}^\star)$, where $\bm{H}(\bm{x}^\star,\bm{y}^\star) = \bm{0}$. If the Jacobian matrix $\nabla_{\bm{y}}\bm{H}(\bm{x}^\star,\bm{y}^\star)^{\mathrm{T}}$ is nonsingular, there exist a neighbourhood $\mathcal{X}$ of $\bm{x}^\star$ in $\mathbb{R}^m$, a neighbourhood $\mathcal{Y}$ of $(\bm{x}^\star,\bm{y}^\star)$ in $\mathbb{R}^{m+n}$, and a mapping $\bm{h} \in C^1(\mathcal{X};\mathbb{R}^n)$ such that $\bm{y} = \bm{h}(\bm{x})$ solves the equation $\bm{H}(\bm{y},\bm{x}) = \bm{0}$ in $\mathcal{Y}$. 
In particular, the implicitly defined mapping $\bm{h}(\cdot)$ is the limit of the sequence $\{\bm{h}_k\}_{ k\in \mathbb{N}_0}^\infty$ of the successive approximations inductively defined by
\begin{subequations}
\begin{align}
\bm{h}_{k+1}(\bm{x}) & = \bm{h}_k(\bm{x}) - \nabla_{\bm{y}}\bm{H}(\bm{x}^\star,\bm{y}^\star)^{-\mathrm{T}}\bm{H}\bigl(\bm{x},\bm{h}_k(\bm{x})\bigr), \\
 \bm{h}_0(\bm{x}) & = \bm{y}^\star,
\end{align}
\end{subequations}
for $\bm{x} \in \mathcal{X}$.
\end{theorem}
In the context of the present paper, the role of Theorem~\ref{thm:Theorem1} is purely analytical and serves to derive a tractable approximation of the implicit friction relationship~\eqref{eq:frModified}. The novelty of the FrSD model does not reside in the application of the implicit-function argument itself, but in the resulting governing equations, which yield the first string-type tyre model with distributed FrBD friction dynamics in the form of semilinear parabolic PDEs. This fundamentally differs from classical brush and stretched-string models, which remain transport dominated and are therefore described by hyperbolic equations. 
In order to proceed and apply Theorem~\ref{thm:Theorem1}, Eq.~\eqref{eq:ODE_Sigma} is first reinterpreted algebraically, that is,
\begin{align}\label{eq:rheol1}
\bm{f}\Biggl(\bm{u},\dpd[2]{\bm{u}}{x}\Biggr) = \mathbf{\Sigma}_0\bm{u} -\mathbf{\Sigma}_1\dpd[2]{\bm{u}}{x}.
\end{align}
Subsequently, adopting the notation of Theorem~\ref{thm:Theorem1} and equating Eqs.~\eqref{eq:rheol1} and~\eqref{eq:frModified} gives
\begin{align}\label{eq:H}
\begin{split}
& \bm{H}\Biggl(\dot{\bm{u}},\bm{u},\dpd[2]{\bm{u}}{x},\bm{v}\ped{r}\Biggr) = \bm{f}\Biggl(\bm{u},\dpd[2]{\bm{u}}{x}\Biggr) - \bm{f}\ped{r}\bigl(\bm{v}\ped{s}(\dot{\bm{u}},\bm{v}\ped{r})\bigr) = \bm{0}, \quad t \in (0,T).
\end{split}
\end{align}
In the sliding regime, where $\norm{\dot{\bm{u}}}_2 \ll \norm{\bm{v}\ped{r}}_2$, Eq.~\eqref{eq:H} may be approximated by invoking Theorem~\ref{thm:Theorem1} with $\bm{x} = (\bm{u},\pd[2]{\bm{u}}{x},\bm{v}\ped{r})$ and $\bm{y} = \dot{\bm{u}} = \od{\bm{u}}{t}$, yielding
\begin{align}\label{eq:zk01}
\begin{split}
\dod{\bm{u}_{k+1}}{t} & = \dod{\bm{u}_k}{t}-\nabla_{\bm{\dot{u}}}\bm{H}\Biggl(\dot{\bm{u}}^\star, \bm{u}^\star,\dpd[2]{\bm{u}}{x}^\star,\bm{v}\ped{r}^\star\Biggr)^{-\mathrm{T}} \bm{H}\Biggl(\dot{\bm{u}}_k,\bm{u}_k,\dpd[2]{\bm{u}_k}{x},\bm{v}\ped{r}\Biggr), \quad k \in \mathbb{N}_0.
\end{split}
\end{align}
In turn, committing the additional approximation \cite{FrBD,2DFrBD,2DFrBD3}
\begin{align}\label{eq:nablaH}
\begin{split}
& \nabla_{\bm{\dot{u}}}\bm{H}\Biggl(\dot{\bm{u}},\bm{u},\dpd[2]{\bm{u}}{x},\bm{v}\ped{r}\Biggr)^{\mathrm{T}}   \approx  \dfrac{\mathbf{M}^2\bigl(\bm{v}\ped{s}(\dot{\bm{u}},\bm{v}\ped{r})\bigr)}{\norm{\mathbf{M}\bigl(\bm{v}\ped{s}(\dot{\bm{u}},\bm{v}\ped{r})\bigr)\bm{v}\ped{s}(\dot{\bm{u}},\bm{v}\ped{r})}_{2,\varepsilon}},
\end{split}
\end{align}
recalling Eq.~\eqref{eq:ODE_Sigma} and~\eqref{eq:rheol1}, and truncating Eq.~\eqref{eq:zk01} at $k=1$ provides, for an initial guess $\dot{\bm{u}}_0(x) = \bm{0}$,
\begin{align}\label{eq:ODEz}
\dot{\bm{u}}(x,t) = -\mathbf{M}^{-2}\bigl(\bm{v}\ped{r}(x,t)\bigr)\norm{\mathbf{M}\bigl(\bm{v}\ped{r}(x,t)\bigr)\bm{v}\ped{r}(x,t)}_{2,\varepsilon}\bm{f}(x,t)-\bm{v}\ped{r}(x,t), \quad (x,t) \in(-a,a)\times (0,T).
\end{align}
The above Eq.~\eqref{eq:ODEz} describes the evolution of the tyre deformations with the rigid relative velocity $\bm{v}\ped{r}(x,t)$ regarded as the input. A state-space representation of Eq.~\eqref{eq:ODEz} is introduced next in Sect.~\ref{sect:dyn}, using the travelled distance as independent time-like variable. Before moving to Sect.~\ref{sect:dyn}, an important consideration is formalised in Remark~\ref{remark:pressure} below.
\begin{remark}\label{remark:pressure}
For a spatially invariant rigid relative velocity $\bm{v}\ped{r}^\star(x) = \bm{v}\ped{r}^\star$, the adoption of a constant pressure distribution makes the arguments of Theorem~\ref{thm:Theorem1} fully rigorous, since, in this case, it is always possible to find an equilbirum of the type\footnote{It should be clarified that the expression reported in Eq.~\eqref{eq:uStar} will not coincide with the steady-state solution of the FrSD tyre model, even for spatially constant rigid relative velocities.}
\begin{align}\label{eq:uStar}
\bm{u}^\star(x) = -\mathbf{\Sigma}_0^{-1}\dfrac{\mathbf{M}^2(\bm{v}\ped{r}^\star)\bm{v\ped{r}^\star}}{\norm{\mathbf{M}(\bm{v}\ped{r}^\star)\bm{v}\ped{r}^\star}_{2,\varepsilon}},
\end{align}
satisfying
\begin{align}
\dpd{\bm{u}^\star(x)}{x} = \dpd[2]{\bm{u}^\star(x)}{x} = \dpd{\bm{u}^\star(x)}{s} = \bm{0},
\end{align}
albeit violating the BCs~\eqref{eq:Bcs}.
However, once Eq.~\eqref{eq:ODEz} has been derived, these unnecessary assumptions may be removed \emph{a posteriori}, which permits postulating spatially varying rigid relative velocities and pressure distributions $p \in C^1([-a,a];\mathbb{R}_{>0})$. It is worth noting, however, that compactly supported profiles, such as the parabolic distribution often employed by the brush models, cannot be directly used due to the singularity of $1/p(x)$ at the leading and trailing edges, and need to be regularised first. Alternatively, this limitation may be overcome by connecting the tyre strings in series with additional bristle elements, as done in \cite{Pauwelussen}. Such an extension is also expected to enhance the overall model accuracy and is planned for implementation in future work.
\end{remark}

\subsubsection{String dynamics}\label{sect:dyn}
When the spin inputs are moderate, the rigid relative velocity $\bm{v}\ped{r}(x,t)$ coincides with the slip velocity $\bm{v}(x,t) \in \mathbb{R}^2$, given by \cite{2DFrBD,2DFrBD3}
\begin{align}
\bm{v}(x,t) =- V\ped{r}(t)\bm{\sigma}(t) -V\ped{r}(t)\varphi(t)\begin{bmatrix} 0\\ x\end{bmatrix},
\end{align}
where $\mathbb{R}^2 \ni \bm{\sigma}(t) = [\sigma_x(t)\; \sigma_y(t)]^{\mathrm{T}}$ and $\varphi(t) \in \mathbb{R}$ denote the theoretical slip and spin inputs, defined as
\begin{align}
\bm{\sigma}(t) & \triangleq -\dfrac{1}{V\ped{r}(t)}\begin{bmatrix} V_{Ox}(t)-V\ped{r}(t)\\ V_{Oy}(t)\end{bmatrix}, && \varphi(t) \triangleq -\dfrac{\dot{\psi}(t)}{V\ped{r}(t)} + \dfrac{1}{R\ped{r}(t)}(1-\varepsilon_\gamma)\sin \bigl(\gamma(t)\bigr),
\end{align}
where $V_{Ox}(t),V_{Oy}(t) \in \mathbb{R}$ denote the longitudinal and lateral speeds of the contact patch centre, respectively, $\dot{\psi}(t) \in \mathbb{R}$ is the turning speed, that is, the angular velocity of the tyre around its vertical axis, $\gamma(t) \in (-\frac{\pi}{2},\frac{\pi}{2})$ the camber angle, and $\varepsilon_\gamma \in [0,1)$ the \emph{camber reduction factor}, which needs to be determined experimentally \cite{Pacejka2,Guiggiani}.

Therefore, replacing the time variable in Eq.~\eqref{eq:ODEz} with the travelled distance $\mathbb{R}_{\geq 0} \ni s \triangleq \int_0^t V\ped{r}(t^\prime) \dif t^\prime$, the governing PDEs of the FrSD tyre model may be inferred as
\begin{subequations}\label{eq:PDEsos}
\begin{align}
\begin{split}
& \dpd{\bm{u}(x,s)}{s} - \mathbf{\Gamma}\bigl(\bar{\bm{v}}(x,s),x,s\bigr)\dpd[2]{\bm{u}(x,s)}{x}  = \dpd{\bm{u}(x,s)}{x}+ \mathbf{\Sigma}\bigl(\bar{\bm{v}}(x,s),x,s\bigr)\bm{u}(x,s) - \bar{\bm{v}}(\bm{x},s), \\
& \qquad \qquad \qquad \qquad \qquad \qquad\qquad \qquad\qquad \qquad \qquad\qquad \qquad \qquad \quad (x,s) \in (-a,a)\times(0,S), \label{es:PDEsos1PDE}
\end{split} \\
&\mathbf{\Lambda}\dpd{\bm{u}(a,s)}{x} + \bm{u}(a,s) = \bm{0}, \label{eq:Bc1}\\
& \mathbf{\Lambda}\dpd{\bm{u}(-a,s)}{x} - \bm{u}(-a,s) = \bm{0}, \quad s\in(0,S), \label{eq:Bc2} \\
& \bm{u}(x,0) = \bm{u}_0(x), \quad x\in(-a,a),
\end{align}
\end{subequations}
where $\mathbb{R}_{>0}\ni S \triangleq \int_0^T V\ped{r}(t) \dif t$, the nondimensional rigid slip velocity is defined as $\mathbb{R}^2 \ni \bar{\bm{v}}(x,s) = [\bar{v}_{x}(x,s)\; \bar{v}_y(x,s)]^{\mathrm{T}} \triangleq \bm{v}(x,s)/V\ped{r}(s)$, and the matrix-valued functions $\mathbf{\Sigma}, \mathbf{\Gamma} : \mathbb{R}^2\times [-a,a]\times \mathbb{R} \mapsto \mathbf{M}_2(\mathbb{R})$ read
\begin{subequations}\label{eq:sigmasX}
\begin{align}
\mathbf{\Sigma}(\bar{\bm{v}},x,s) &  \triangleq \mathbf{\Psi}(\bar{\bm{v}},s)\mathbf{\Sigma}_0(x), \\
\mathbf{\Gamma}(\bar{\bm{v}},x,s) & \triangleq -\mathbf{\Psi}(\bar{\bm{v}},s)\mathbf{\Sigma}_1(x), \label{eq:Gamma}
\end{align}
\end{subequations}
with $\mathbf{\Psi} : \mathbb{R}^2 \times \times \mathbb{R}_{\geq 0} \mapsto \mathbf{M}_2(\mathbb{R})$ given by
\begin{align}
\mathbf{\Psi}(\bar{\bm{v}},s) \triangleq -\dfrac{1}{V\ped{r}(s)}\mathbf{M}^{-2}\bigl(V\ped{r}(s)\bar{\bm{v}}\bigr)\norm{\mathbf{M}\bigl(V\ped{r}(s)\bar{\bm{v}}\bigr)V\ped{r}(s)\bar{\bm{v}}}_{2,\varepsilon}.
\end{align}
Some considerations about the mathematical structure of the above PDEs~\eqref{eq:PDEsos} are in order. First, in contrast to most rolling contact models employing Coulomb or FrBD-type friction, which are formulated as systems of \emph{hyperbolic} PDEs, Eq.~\eqref{eq:PDEsos} is \emph{parabolic} for all $\varepsilon \in \mathbb{R}_{>0}$. This property follows directly from the constitutive assumption~\eqref{eq:ODE_Sigma}, which, by relating the string deformations to the shear stresses via an elliptic operator \cite{Pauwelussen}, modifies the underlying hyperbolic structure of classical rolling contact dynamics. This appears to be a new finding and demonstrates that, by injecting higher-order partial derivatives into their governing PDEs, friction can fundamentally alter the mathematical character of rolling contact processes.
It is worth noting that, in the limiting case $\bar{\bm{v}}(x,s)=\bm{0}$ and $\varepsilon=0$, the matrix $\mathbf{\Gamma}(\bar{\bm{v}}(x,s),x,s)$ defined in Eq.~\eqref{eq:Gamma} vanishes identically, and Eq.~\eqref{eq:PDEsos} reverts to a hyperbolic system. This degenerate behaviour can be effectively avoided by introducing a strictly positive regularisation parameter $\varepsilon \in \mathbb{R}_{>0}$ in Eq.~\eqref{eq:frModified}. Besides restoring parabolicity for all inputs, this regularisation enables the establishment of well-posedness results (see, e.g., Sect.~\ref{sect:wellP}) and facilitates numerical implementation by suppressing spurious instabilities.

Second, beyond improving numerical robustness, the regularisation parameter $\varepsilon$ allows both boundary conditions~\eqref{eq:Bc1} and~\eqref{eq:Bc2} to be imposed simultaneously. This yields a smooth transition between the free portions of the tyre and the contact region, which is physically sound and consistent with empirical evidence, as tread and carcass deformations are known to exhibit smooth spatial profiles. By contrast, owing to their intrinsically hyperbolic nature, classical string models with Coulomb friction must necessarily relax one boundary condition (typically enforcing only Eq.~\eqref{eq:Bc1}, which leads to the well-known \emph{no-kink} condition at the leading edge).

Third, since the matrix $\mathbf{\Gamma}(\bar{\bm{v}}(x,s),x,s)$ in Eq.~\eqref{eq:Gamma} depends explicitly on slip and spin inputs, the diffusive term in Eq.~\eqref{es:PDEsos1PDE} introduces relaxation effects that are directly modulated by the magnitude of these quantities. This feature is again in agreement with experimental observations, which indicates progressively shorter transients as slip levels increase.

\subsubsection{Tyre forces and moments}\label{sect:forces}
Starting with Eqs.~\eqref{eq:PDEsos} and~\eqref{eq:ODE_SC}, the tangential forces and the aligning moment acting on the tyre, $\mathbb{R}^2\ni \bm{F}_{\bm{t}}(s) = [F_{x}(s)\; F_y(s)]^{\mathrm{T}}$ and $M_z(s) \in \mathbb{R}$, respectively, may be determined by integration over the contact length, according to
\begin{subequations}\label{eq:FandM}
\begin{align}
\bm{F}_{\bm{t}}(s) & = \int_{-\infty}^\infty \bm{q}(x,s) \dif x = \int_{-a}^a \bm{q}(x,s) \dif x, \label{eq:Fundef}\\
\begin{split}
M_z(s) & = \int_{-\infty}^\infty \bigl[ xq_y(x,s)-u_y(x,s)q_x(x,s)\bigr] \dif x \\
& = \int_{-a}^a \bigl[ xq_y(x,s)-u_y(x,s)q_x(x,s)\bigr] \dif x, \quad s \in [0,S]. \label{eq:Mzunderfr}
\end{split}
\end{align}
\end{subequations}
It should be emphasised that a term $u_x(x,s)q_y(x,s)$ has been neglected in the computation of the aligning moment in Eq.~\eqref{eq:Mzunderfr}, which may be justified by observing that the deformation of the tyre tread is usually small compared to the longitudinal coordinate. Conversely, the contribution $u_y(x,s)q_x(x,s)$ must be included due to the relatively large deflections undergone by the tyre carcass.

Equations~\eqref{eq:FandM} may be recast in a more explicit form that is also amenable to mathematical analysis. Specifically, integrating by parts and enforcing the BCs~\eqref{eq:Bc1} and~\eqref{eq:Bc2} yields
\begin{subequations}\label{eq:FMineg}
\begin{align}\label{eq:FMinegF}
\begin{split}
\bm{F}_{\bm{t}}(s) & = \int_{-a}^{a} \mathbf{K}(x)\bm{u}(x,s) \dif x + \mathbf{S}\mathbf{\Lambda}^{-1}\bm{u}(a,s) + \mathbf{S}\mathbf{\Lambda}^{-1}\bm{u}(-a,s),
\end{split}\\
\begin{split}
M_z(s) & = \int_{-a}^a xk_yu_y(x,s) \dif x + \Biggl( a\dfrac{S}{\lambda_y} + S\Biggr)u_y(a,s) - \Biggl( a\dfrac{S}{\lambda_y} + S\Biggr)u_y(-a,s) -\int_{-a}^{a}u_y(x,s)k_xu_x(x,s)\dif x \\
& \quad -u_y(a,s)\dfrac{EA}{\lambda_x}u_x(a,s)-u_y(-a,s)\dfrac{EA}{\lambda_x}u_x(-a,s) - \int_{-a}^{a}EA\dpd{u_y(x,s)}{x}\dpd{u_y(x,s)}{x}\dif x, \quad s \in [0,S].
\end{split}
\end{align}
\end{subequations}
Equations~\eqref{eq:PDEsos} and~\eqref{eq:FMineg} completely characterise the FrSD tyre model. Its mathematical properties are studied next in Sect.~\ref{sect:math}.

\section{Mathematical properties}\label{sect:math}
The present section is devoted to the mathematical analysis of the FrSD model described by Eqs.~\eqref{eq:PDEsos} and~\eqref{eq:FMineg}. More specifically, Sect.~\ref{sect:wellP} investigates well-posedness properties, whereas stability and passivity constitute the main subject of Sect.~\ref{sect:stabAndPass}.

\subsection{Well-posedness}\label{sect:wellP}
Existence and uniqueness properties for the solution of the PDE~\eqref{eq:PDEsos} may be proved within different functional settings. For simplicity, the following analysis focuses on sufficiently regular solutions, for which well-posedness can be established by invoking standard results already available from the literature.
To this end, it is profitable to define the following functions:
\begin{subequations}\label{eq:funTilde}
\begin{align}
\tilde{\mathbf{\Sigma}}(x,s) & \triangleq \mathbf{\Sigma}\bigl(\bar{\bm{v}}(x,s),x,s\bigr), \\
\tilde{\mathbf{\Gamma}}(x,s) & \triangleq \mathbf{\Gamma}\bigl(\bar{\bm{v}}(x,s),x,s\bigr), \\
\tilde{\bm{h}}(x,s) & \triangleq \bar{\bm{v}}(x,s),
\end{align}
\end{subequations}
so that the PDE~\eqref{eq:PDEsos} may be recast as
\begin{subequations}\label{eq:PDerecast}
\begin{align}
\begin{split}
& \dpd{\bm{u}(x,s)}{s} - \tilde{\mathbf{\Gamma}}(x,s)\dpd[2]{\bm{u}(x,s)}{x}  = \dpd{\bm{u}(x,s)}{x}+ \tilde{\mathbf{\Sigma}}(x,s)\bm{u}(x,s) + \tilde{\bm{h}}(x,s),\quad (x,s) \in (-a,a)\times(0,S),
\end{split} \\
&\mathbf{\Lambda}\dpd{\bm{u}(a,s)}{x} + \bm{u}(a,s) = \bm{0}, \label{eq:Bc1tilde}\\
& \mathbf{\Lambda}\dpd{\bm{u}(-a,s)}{x} - \bm{u}(-a,s) = \bm{0}, \quad s\in(0,S), \label{eq:Bc2tilde} \\
& \bm{u}(x,0) = \bm{u}_0(x), \quad x\in(-a,a).
\end{align}
\end{subequations}
The more compact representation~\eqref{eq:PDerecast} permits recovering well-posedness invoking the results contained in \cite{Lunardi}. In particular, for every $\alpha \in (0,1)$, the following spaces are conveniently defined:
\begin{subequations}
\begin{align}
C^{\alpha,0}([-a,a]\times[0,S];\mathbb{R}^2) & \triangleq \Bigl\{ \bm{f} \in C^0([-a,a]\times[0,S];\mathbb{R}^2) \mathrel{\Big |} \bm{f}(\cdot,s) \in C^\alpha([-a,a];\mathbb{R}^2), \; s \in [0,S]\Bigr\},  \\
C^{2,1}([-a,a]\times[0,S];\mathbb{R}^2) & \triangleq \Biggl\{\bm{f} \in C^0([-a,a]\times[0,S];\mathbb{R}^2) \mathrel{\Bigg |} \dpd{\bm{f}}{t}, \dpd[2]{\bm{f}}{x} \in C^0([-a,a]\times[0,S];\mathbb{R}^2) \Biggr\}, \\
C^{2,1+\alpha}([-a,a]\times[0,S];\mathbb{R}^2) & \triangleq \Biggl\{\bm{f} \in C^{2,1}([-a,a]\times[0,S];\mathbb{R}^2) \mathrel{\Bigg |} \dpd{\bm{f}}{t}, \dpd[2]{\bm{f}}{x} \in C^{\alpha,0}([-a,a]\times[0,S];\mathbb{R}^2) \Biggr\}.
\end{align}
\end{subequations}
Theorem~\ref{thm:ex1} enounces existence and uniqueness results for regular solutions of Eq.~\eqref{eq:PDerecast}.
\begin{theorem}[Existence and uniqueness of regular solutions]\label{thm:ex1}
For all $\tilde{\mathbf{\Sigma}} \in C^1([-a,a]\times[0,S];\mathbf{M}_{2}(\mathbb{R}))$ and $\tilde{\bm{h}} \in C^1([-a,a]\times[0,S];\mathbb{R}^{2})$ as in Eq.~\eqref{eq:funTilde}, and ICs $\bm{u}_0 \in C^{2+\alpha}([-a,a];\mathbb{R}^2)$, $\alpha\in (0,1)$, satisfying the BCs~\eqref{eq:Bc1tilde} and~\eqref{eq:Bc2tilde}, the PDE~\eqref{eq:PDerecast} admits a unique \emph{regular solution} $\bm{u} \in C^{2,1+\alpha}([-a,a]\times[0,S];\mathbb{R}^2)$ that also satisfies the BCs~\eqref{eq:Bc1tilde} and~\eqref{eq:Bc2tilde}.
\begin{proof}
See Theorem 5.1.21 in \cite{Lunardi}.
\end{proof}
\end{theorem}
By composition of continuously differentiable functions, it may be realised that $\mathbf{M}\in C^1(\mathbb{R}^2; \mathbf{Sym}_2(\mathbb{R}))$ implies $\tilde{\mathbf{\Sigma}} \in C^1([-a,a]\times[0,S];\mathbf{M}_2(\mathbb{R}))$ and $\tilde{\bm{h}} \in C^1([-a,a] \times [0,S];\mathbb{R}^2)$ for all $\bar{\bm{v}} \in C^1([-a,a]\times[0,S];\mathbb{R}^2)$ and $V\ped{r} \in C^1([0,S]; [V\ped{min}, V\ped{max}])$. Essentially, the PDE~\eqref{eq:PDerecast}  with $\varepsilon \in \mathbb{R}_{>0}$ admits regular solutions for sufficiently smooth slip inputs and rolling velocities.

\subsection{Stability and passivity}\label{sect:stabAndPass}
Stability and passivity are two fundamental mathematical properties of the PDE~\eqref{eq:PDerecast}. The former ensures that both the tyre deformations and the resulting forces and moments remain bounded in time (or, equivalently, with respect to the travelled distance), whilst the latter guarantees that energy is dissipated through frictional rolling contact. Both properties can be analysed by introducing the following Lyapunov function, which represents the total elastic energy stored in the tyre:
\begin{align}\label{eq:Lyap}
\begin{split}
W\bigl(\bm{u}(\cdot,s)\bigr) & = \dfrac{1}{2}\int_{-a}^{a} \Biggl(\bm{u}^{\mathrm{T}}(x,s)\mathbf{K}\bm{u}(x,s) + \dpd{\bm{u}^{\mathrm{T}}(x,s)}{x}\mathbf{S}\dpd{\bm{u}(x,s)}{x}\Biggr) \dif x  \\
& \quad + \dfrac{1}{2}\bm{u}^{\mathrm{T}}(a,s)\mathbf{S}\mathbf{\Lambda}^{-1}\bm{u}(a,s) + \dfrac{1}{2}\bm{u}^{\mathrm{T}}(-a,s)\mathbf{S}\mathbf{\Lambda}^{-1}\bm{u}(-a,s).
\end{split}
\end{align}
For regular solutions, as those considered in Theorem~\ref{thm:ex1}, it follows from Sobolev embedding in one dimension that $\norm{\bm{u}(a,s)}_2, \norm{\bm{u}(-a,s)}_2 \leq C\norm{\bm{u}(\cdot,s)}_{H^1((-a,a);\mathbb{R}^2)}$ for some $C \in \mathbb{R}_{>0}$, and therefore the above Lyapunov function is equivalent to a squared norm on $H_1((-a,a);\mathbb{R}^2)$, and is also differentiable in the classical sense.
Accordingly, taking the derivative of Eq.~\eqref{eq:Lyap} along the dynamics~\eqref{eq:PDerecast} gives
\begin{align}\label{eq;Lyap1diff0}
\begin{split}
\dod{W(s)}{s} & = \int_{-a}^{a}\bm{u}^{\mathrm{T}}(x,s)\mathbf{K}\dpd{\bm{u}(x,s)}{s} \dif x + \int_{-a}^a \dpd{\bm{u}^{\mathrm{T}}(x,s)}{x}\mathbf{S}\dmd{\bm{u}(x,s)}{2}{x}{}{s}{} \dif x\\
&  \quad + \bm{u}^{\mathrm{T}}(a,s)\mathbf{S}\mathbf{\Lambda}^{-1}\dod{\bm{u}(a,s)}{s} + \bm{u}^{\mathrm{T}}(-a,s)\mathbf{S}\mathbf{\Lambda}^{-1}\dod{\bm{u}(-a,s)}{s}, \quad s \in (0,S). 
\end{split}
\end{align}
Integrating by parts the second term on the right-hand side of Eq.~\eqref{eq;Lyap1diff0} provides 
\begin{align}
\begin{split}
\int_{-a}^a \dpd{\bm{u}^{\mathrm{T}}(x,s)}{x}\mathbf{S}\dmd{\bm{u}(x,s)}{2}{x}{}{s}{} \dif x &= \dpd{\bm{u}^{\mathrm{T}}(a,s)}{x}\mathbf{S}\dod{\bm{u}(a,s)}{s} - \dpd{\bm{u}^{\mathrm{T}}(-a,s)}{x}\mathbf{S}\dod{\bm{u}(-a,s)}{s} \\
& \quad - \dpd[2]{\bm{u}^{\mathrm{T}}(x,s)}{x}\mathbf{S}\dpd{\bm{u}(x,s)}{s} \dif x,
\end{split}
\end{align}
which, inserted into Eq.~\eqref{eq;Lyap1diff0}, yields
\begin{align}\label{eq;Lyap1diff}
\begin{split}
\dod{W(s)}{s} & = \int_{-a}^{a}\Biggl(\mathbf{K}\bm{u}(x,s)-\mathbf{S}\dpd[2]{\bm{u}(x,s)}{x}\Biggr)^{\mathrm{T}}\dpd{\bm{u}(x,s)}{s} \dif x +\biggl( \dpd{\bm{u}(a,s)}{x} + \mathbf{\Lambda}^{-1}\bm{u}(a,s)\biggr)^{\mathrm{T}}\mathbf{S}\dod{\bm{u}(a,s)}{s}\\
&  \quad - \biggl( \dpd{\bm{u}(-a,s)}{x} - \mathbf{\Lambda}^{-1}\bm{u}(-a,s)\biggr)^{\mathrm{T}}\mathbf{S}\dod{\bm{u}(-a,s)}{s}, \quad s \in (0,S). 
\end{split}
\end{align}
Enforcing the BCs~\eqref{eq:Bc1tilde} and~\eqref{eq:Bc2tilde} and substituting Eq.~\eqref{es:PDEsos1PDE} into~\eqref{eq;Lyap1diff} leads then to
\begin{align}\label{eq:lyap2}
\begin{split}
\dod{W(s)}{s} & = -\bigl\langle \bm{q}(\cdot,s), \bar{\bm{v}}(\cdot,s)\bigr\rangle_{L^2((-a,a);\mathbb{R}^2)} \\
& \quad + \int_{-a}^{a}\Biggl(\mathbf{K}\bm{u}(x,s)-\mathbf{S}\dpd[2]{\bm{u}(x,s)}{x}\Biggr)^{\mathrm{T}}\dfrac{\mathbf{\Psi}\bigl(\bar{\bm{v}}(\bm{x},s),s) }{p(x)} \Biggl(\mathbf{K}\bm{u}(x,s)-\mathbf{S}\dpd[2]{\bm{u}(x,s)}{x}\Biggr) \dif x \\
& \quad +\int_{-a}^{a}\Biggl(\mathbf{K}\bm{u}(x,s)-\mathbf{S}\dpd[2]{\bm{u}(x,s)}{x}\Biggr)^{\mathrm{T}}\dpd{\bm{u}(x,s)}{x} \dif x,   \quad s \in (0,S). 
\end{split}
\end{align}
Integrating by parts the last term appearing in Eq.~\eqref{eq:lyap2} gives
\begin{align}
\begin{split}
& \int_{-a}^{a}\Biggl(\mathbf{K}\bm{u}(x,s)-\mathbf{S}\dpd[2]{\bm{u}(x,s)}{x}\Biggr)^{\mathrm{T}}\dpd{\bm{u}(x,s)}{x} \dif x = \dfrac{1}{2}\bm{u}^{\mathrm{T}}(a,s)\mathbf{K}\bm{u}(a,s) - \dfrac{1}{2}\dpd{\bm{u}^{\mathrm{T}}(a,s)}{x}\mathbf{S}\dpd{\bm{u}(a,s)}{x} \\
& \quad -\dfrac{1}{2}\bm{u}^{\mathrm{T}}(-a,s)\mathbf{K}\bm{u}(-a,s) +\dfrac{1}{2} \dpd{\bm{u}^{\mathrm{T}}(-a,s)}{x}\mathbf{S}\dpd{\bm{u}(-a,s)}{x},
\end{split}
\end{align}
and imposing again the BCs~\eqref{eq:Bc1tilde} and~\eqref{eq:Bc2tilde}, Eq.~\eqref{eq:lyap2} simplifies to
\begin{align}\label{eq:lyap3}
\begin{split}
\dod{W(s)}{s} & = -\bigl\langle \bm{q}(\cdot,s), \bar{\bm{v}}(\cdot,s)\bigr\rangle_{L^2((-a,a);\mathbb{R}^2)} \\
& \quad + \int_{-a}^{a}\Biggl(\mathbf{K}\bm{u}(x,s)-\mathbf{S}\dpd[2]{\bm{u}(x,s)}{x}\Biggr)^{\mathrm{T}}\dfrac{\mathbf{\Psi}\bigl(\bar{\bm{v}}(\bm{x},s),s) }{p(x)} \Biggl(\mathbf{K}\bm{u}(x,s)-\mathbf{S}\dpd[2]{\bm{u}(x,s)}{x}\Biggr) \dif x,   \quad s \in (0,S). 
\end{split}
\end{align}
The inequality~\eqref{eq:lyap3} will be instrumental in proving stability and passivity properties for the FrSD model. The next Sect.~\ref{sect:stab} focuses specifically on stability.

\subsubsection{Stability}\label{sect:stab}
In the following, stability is analysed in terms of both state and input-to-output behaviour of the PDE~\eqref{eq:PDerecast}, by adopting the notions of \emph{input-to-state} and \emph{input-to-output} stability. The former characterises the boundedness of the distributed states, and thus of the tyre deformations, whereas the latter concerns the boundedness of the generated forces and moment.
In this context, the first theoretical result is formalised in Lemma~\ref{lemma:stab}.

\begin{lemma}[Input-to-state stability (ISS)]\label{lemma:stab}
For all inputs $\bar{\bm{v}} \in C^1(\mathscr{P}\times\mathbb{R}_{\geq 0};\mathbb{R}^2) \cap L^\infty([-a,a]\times\mathbb{R}_{\geq 0};\mathbb{R}^2)$, rolling speeds $V\ped{r} \in C^1(\mathbb{R}_{\geq 0};[V\ped{min},V\ped{max}])$, and ICs $\bm{u}_0 \in C^{2+\alpha}([-a,a];\mathbb{R}^2)$, $\alpha \in (0,1)$, satisfying the BCs~\eqref{eq:Bc1tilde} and~\eqref{eq:Bc2tilde}, the PDE~\eqref{eq:PDerecast} is (uniformly) \emph{input-to-state stable} (ISS) in the spatial $H^1$-norm. In particular, there exist functions $\beta \in \mathcal{KL}$ and $\gamma \in \mathcal{K}_\infty$ such that
\begin{align}\label{eq:ISSbetaGamma}
\begin{split}
\norm{\bm{u}(\cdot,s)}_{H^1((-a,a);\mathbb{R}^2)} & \leq \beta\Bigl(\norm{\bm{u}_0(\cdot)}_{H^1((-a,a);\mathbb{R}^2)}, s\Bigr) + \gamma\Bigl( \norm{\bar{\bm{v}}(\cdot,\cdot)}_\infty\Bigr), \quad s \in \mathbb{R}_{\geq 0}.
\end{split}
\end{align}
\begin{proof}
Integrating by parts, it may be easily inferred that
\begin{align}\label{eq:normEq}
\begin{split}
& \norm{\bm{q}(\cdot,s)}_{L^2((-a,a);\mathbb{R}^2)}^2 = \int_{-a}^{a} \Biggl(\mathbf{K}\bm{u}(x,s)-\mathbf{S}\dpd[2]{\bm{u}(x,s)}{x}\Biggr)^{\mathrm{T}}\Biggl(\mathbf{K}\bm{u}(x,s)-\mathbf{S}\dpd[2]{\bm{u}(x,s)}{x}\Biggr)\dif x \\
& \quad \geq \int_{-a}^{a}\biggl(\bm{u}^{\mathrm{T}}(x,s)\mathbf{K}^2\bm{u}(x,s) +2 \dpd{\bm{u}^{\mathrm{T}}(x,s)}{x}\mathbf{K}\mathbf{S}\dpd{\bm{u}(x,s)}{x} +\dpd[2]{\bm{u}^{\mathrm{T}}(x,s)}{x}\mathbf{S}^2\dpd[2]{\bm{u}(x,s)}{x}\biggr)\dif x \\
& \quad \geq \gamma_0\norm{\bm{u}(\cdot,s)}_{H^1((-a,a);\mathbb{R}^2)}^2,
\end{split}
\end{align}
for some $\gamma_0 \in \mathbb{R}_{>0}$. Moreover, $\varepsilon \in \mathbb{R}_{>0}$ implies the existence of $\gamma_1 \in \mathbb{R}_{>0}$ such that
\begin{align}
\int_{-a}^{a}\Biggl(\mathbf{K}\bm{u}(x,s)-\mathbf{S}\dpd[2]{\bm{u}(x,s)}{x}\Biggr)^{\mathrm{T}}\dfrac{\mathbf{\Psi}\bigl(\bar{\bm{v}}(\bm{x},s),s) }{p(x)} \Biggl(\mathbf{K}\bm{u}(x,s)-\mathbf{S}\dpd[2]{\bm{u}(x,s)}{x}\Biggr) \dif x \leq \gamma_1 \norm{\bm{q}(\cdot,s)}_{L^2((-a,a);\mathbb{R}^2)}^2.
\end{align}
Consequently, applying Cauchy-Schwarz and then the generalised Young's inequality for product to the first term in Eq.~\eqref{eq:lyap3} yields 
\begin{align}
\dod{W(s)}{s} \leq - \dfrac{\gamma_1}{2}\norm{\bm{q}(\cdot,s)}_{L^2((-a,a);\mathbb{R}^2)}^2 + \dfrac{1}{2\gamma_1}\norm{\bar{\bm{v}}(\cdot,s)}_{L^2((-a,a);\mathbb{R}^2)}, \quad s \in \mathbb{R}_{>0}.
\end{align}
Recalling Eq.~\eqref{eq:normEq}, it may be concluded that there exist $\rho, \mu \in \mathbb{R}_{>0}$ such that
\begin{align}
\dod{W(s)}{s} \leq -\rho W(s) + \mu\norm{\bar{\bm{v}}(\cdot,\cdot)}_\infty^2, \quad s \in \mathbb{R}_{>0}.
\end{align}
Finally, invoking the Grönwall-Bellman inequality provides
\begin{align}\label{eq:estFinal}
W\bigl(\bm{u}(\cdot,s)\bigr) \leq \exp(-\rho s)W\bigl(\bm{u}_0(\cdot)\bigr) + \dfrac{\mu}{\rho}\mu\norm{\bar{\bm{v}}(\cdot,\cdot)}_\infty^2, \quad s \in \mathbb{R}_{\geq 0}.
\end{align}
Since the Lyapunov function~\eqref{eq:Lyap} is equivalent to a squared norm on $H^1((-a,a);\mathbb{R}^2)$, Eq.~\eqref{eq:estFinal} implies~\eqref{eq:ISSbetaGamma}.
\end{proof}
\end{lemma}
Lemma~\ref{lemma:stab} is concerned with the ISS behaviour of the PDE~\eqref{eq:PDerecast}, and asserts boundedness for the distributed state $\bm{u}(x,s)$ in the spatial $H^1$-norm. Combining this result with Eq.~\eqref{eq:FMineg}, it is possible to deduce the existence of $\beta_1, \beta_2, \beta_3 \in \mathbb{R}_{>0}$ such that
\begin{subequations}\label{eq:FyMbound}
\begin{align}
\norm{\bm{F}_{\bm{t}}(s)}_2 & \leq \beta_1\norm{\bm{u}(\cdot,s)}_{H^1((-a,a);\mathbb{R}^2)}, \\
\abs{M_z(s)} & \leq \beta_2\norm{\bm{u}(\cdot,s)}_{H^1((-a,a);\mathbb{R}^2)} + \beta_3\norm{\bm{u}(\cdot,s)}_{H^1((-a,a);\mathbb{R}^2)}^2, \quad s \in \mathbb{R}_{\geq 0}.
\end{align}
\end{subequations} 
The bounds derived in Eq.~\eqref{eq:FyMbound} immediately imply IOS for the tyre forces and aligning moment, according to Corollary~\ref{corollary:IOS} below.
\begin{corollary}[Input-to-output stability (IOS)]\label{corollary:IOS}
Under the same assumptions as Lemma~\ref{lemma:stab}, the PDE~\eqref{eq:PDerecast} with output~\eqref{eq:FMineg} is \emph{input-to-output stable} (IOS).
\end{corollary}
Corollary~\ref{corollary:IOS} concludes the stability analysis of the FrSD model. Before moving to Sect.~\ref{sect:pass}, an additional remark is collected below.

\begin{remark}
Stability was proved by explicitly exploiting the presence of the regularisation parameter $\varepsilon \in \mathbb{R}_{>0}$ in the diffusive term. This choice was motivated by the observation that its introduction appears to be an indispensable requirement for preserving the parabolic character of Eq.~\eqref{eq:PDEsos} for all slip and spin inputs. Nevertheless, it is worth noting that alternative notions of stability can be proved without relying on its presence. For instance, in the case of a spatially invariant rigid slip velocity (i.e., in the absence of spin), boundedness can be established using exactly the same arguments as those in \cite{FrBDvisc}.
\end{remark}

\subsubsection{Dissipativity and passivity}\label{sect:pass}
Being primarely governed by friction, tyre-road interactions dissipate energy. Concerning the string-like tyre model detailed in Sect.~\ref{sect:constitutibe}, this is obviously true when sliding occurs, and the nondimensional friction force assumes the form~\eqref{eq:frModified}. However, since the FrSD model described by Eqs.~\eqref{eq:PDEsos} and~\eqref{eq:FandM} is merely an approximation to the true string-like tyre dynamics, it is interesting to investigate whether the dissipative properties of the rolling process are preserved after applying Theorem~\ref{thm:Theorem1}. In particular, recalling Eq.~\eqref{eq:slidingvel} and computing the dissipation rate gives
\begin{align}\label{eq:frictionPass}
\begin{split}
\bigl\langle \bm{q}(\cdot,s),\bm{v}\ped{s}(\cdot,s)\bigr\rangle_{L^2((-a,a);\mathbb{R}^2)} & = -\int_{-a}^{a} \bm{q}^{\mathrm{T}}(x,s)\bm{v}\ped{s}(x,s) \dif x \\
& =- \int_{-a}^{a} V\ped{r}(s)\bm{q}^{\mathrm{T}}(x,s)\dfrac{\mathbf{\Psi}\bigl(\bar{\bm{v}}\ped{r}(x,s),s\bigr)}{p(x)}\bm{q}(x,s) \dif x \geq 0, \quad s \in [0,S],
\end{split}
\end{align}
which confirms that the FrSD model still predicts energy dissipation within the tyre's contact patch.
The inequality derived according to Eq.~\eqref{eq:frictionPass} employs the actual sliding velocity of the tyre particles, and pertains to the physical nature of the rolling contact process. From a dynamical system theory perspective, however, it is also meaningful to study the passivity properties of the FrSD model when the nondimensional rigid slip velocity $\bar{\bm{v}}(x,s)$ is regarded as the input to the PDE~\eqref{eq:PDerecast} with output~\eqref{eq:FMinegF}. In this context, passivity results are formalised in Lemma~\ref{lemma:pass} below.
\begin{lemma}[Passivity]\label{lemma:pass}
For all inputs $\bar{\bm{v}} \in C^1([-a,a]\times\mathbb{R}_{\geq 0};\mathbb{R}^2) \cap L^\infty(\mathscr{P}\times\mathbb{R}_{\geq 0};\mathbb{R}^2)$, rolling speeds $V\ped{r} \in [V\ped{min},V\ped{max}]$, and ICs $\bm{u}_0 \in C^{2+\alpha}([-a,a];\mathbb{R}^2)$, $\alpha \in (0,1)$, satisfying the BCs~\eqref{eq:Bc1tilde} and~\eqref{eq:Bc2tilde}, the system described by the PDE~\eqref{eq:PDerecast} with output~\eqref{eq:FMinegF} is passive. In particular,
\begin{align}\label{eq:Fvres}
\begin{split}
&-\int_0^s \bigl\langle\bm{q}(\cdot,s^\prime),\bar{\bm{v}}(\cdot,s^\prime)\bigr\rangle_{L^2((-a,a);\mathbb{R}^2)} \dif s^\prime \geq W\bigl(\bm{u}(\cdot,s)\bigr)-W\bigl(\bm{u}_0(\cdot)\bigr), \quad s \in [0,S],
\end{split}
\end{align}
with storage function $W(\bm{u}(\cdot,s))$ as in Eq.~\eqref{eq:Lyap}.
\begin{proof}
Since the second term on the right-hand side of Eq.~\eqref{eq:lyap3} is negative definite, it holds that
\begin{align}\label{eq:passDiff}
-\bigl\langle \bm{q}(\cdot,s), \bar{\bm{v}}(\cdot,s)\bigr\rangle_{L^2((-a,a);\mathbb{R}^2)} \geq \dod{W\bigl(\bm{u}(\cdot,s)\bigr)}{s}, \quad s \in (0,S).
\end{align}
Integrating the above~\eqref{eq:passDiff} yields the result.
\end{proof}
\end{lemma}
Combining Eqs.~\eqref{eq:lyap3} and~\eqref{eq:frictionPass}, it may also be concluded that
\begin{align}
-\bigl\langle \bm{q}(\cdot,s), \bm{v}(\cdot,s)\bigr\rangle_{L^2((-a,a);\mathbb{R}^2)} = -\bigl\langle \bm{q}(\cdot,s), \bm{v}\ped{s}(\cdot,s)\bigr\rangle_{L^2((-a,a);\mathbb{R}^2)} + V\ped{r}(s)\dod{W\bigl(\bm{u}(\cdot,s)\bigr)}{s}, \quad s \in [0,S],
\end{align}
which establishes that the difference between the energy dissipation rates associated with the sliding and rigid slip velocities is precisely equal to the elastic energy stored in the tyre system. In line with the findings of \cite{TransientLosses}, this further implies that slip losses cannot, in general, be evaluated as the inner product of the global tyre forces and moment with the corresponding slip and spin inputs, with the equivalence being only recovered in steady-state conditions \cite{Toyota1,Toyota2,TyreLosses}.

\section{Numerical and experimental validation}\label{sect:NumAnDecp}
The FrSD tyre model developed in Sect.~\ref{sect:FrBD} is validated numerically and experimentally in Sects.~\ref{sect:numerical} and~\ref{sect:exp}.

\subsection{Numerical simulations}\label{sect:numerical}
In the following, simulation results are reported to illustrate both the stationary and transient behaviours of the tyre predicted using the FrSD model. Specifically, Sect.~\ref{sect:steady} is dedicated to the qualitative analysis of the steady-state tyre forces and moments, whereas the dynamic features of the new formulation are discussed in Sect.~\ref{sect:trans}.

\subsubsection{Steady-state behaviour}\label{sect:steady}
A first set of numerical results is presented to elucidate the steady-state behavior of the FrSD tyre model developed in Sect.~\ref{sect:FrBD}. Figure~\ref{fig:String} shows the lateral deflection $u_y(x)$, lateral stress $q_y(x)$, and longitudinal stress $q_x(x)$ for a translational slip input $\bm{\sigma} = (0.2,0.2)$, considering two model parametrisations (P1 and P2) with parameters listed in Table~\ref{tab:parameters} and taken from \cite{Higuchi1}.

Specifically, Fig.~\ref{fig:String}(a), corresponding to P1 in Table~\ref{tab:parameters}, highlights the portions of the string in contact with the road (shaded blue) and the free segments of the tyre (shaded green). Consistent with the model assumptions and imposed boundary conditions, the lateral deflection $u_y(x)$ decays exponentially for $x \in (-\infty,-a] \cup [a,\infty)$, whilst displaying an almost parabolic profile within the contact region $x \in [-a,a]$. Notably, no kink appears at the trailing edge, in agreement with the discussion initiated in Sect.~\ref{sect:dyn}. 
Both the longitudinal and lateral stresses vanish outside the contact zone and grow nonlinearly from the leading to the trailing edge. In particular, the longitudinal stresses reach higher magnitudes than lateral ones, reflecting the tyre's greater stiffness along the longitudinal direction.
Analogous observations apply to the trends illustrated in Fig.~\ref{fig:String}(b), produced employing smaller values for the relaxation lengths (P2 in Table~\ref{tab:parameters}). The predicted behavior is qualitatively consistent with that reported in \cite{Higuchi1}, with the main difference being that the shear stresses remain continuous along the contact length.

Figures~\ref{fig:tirePlots} and~\ref{fig:tirePlots2} illustrate the main tyre characteristics produced employing parametrisations P1 and P2, respectively, and considering different combinations of translational slip inputs and a constant pressure distribution. In agreement with classical results from Coulomb-type friction models, the lateral tyre force $F_y$ and the aligning moment $M_z$ exhibit a peak whose magnitude decreases under combined slip conditions. Both quantities approach asymptotic values at large $\sigma_y$. Although not explicitly shown in Figs.~\ref{fig:tirePlots} and~\ref{fig:tirePlots2}, the sign reversal phenomenon typically observed in the aligning moment can be reproduced using asymmetric pressure distributions, for example, a slowly decaying exponential along the contact length.
Inspection of the friction ellipses in Figs.~\ref{fig:tirePlots} and~\ref{fig:tirePlots2} indicate that, despite the applied normal load of $F_z = 3$ kN, the total tangential force does not exceed 2.5 kN at very large slips. This behaviour is a direct consequence of the adopted friction model, in which the dynamic coefficient is significantly lower than the static one (Table~\ref{tab:parameters}). Finally, particularly noteworthy is the trend of the aligning moment versus the longitudinal force. Large slip values induce significant carcass deflections, producing asymmetries in the vertical moment, as evident in both figures. In particular, negative longitudinal forces lead to a sign reversal, in agreement with the predictions of \cite{Higuchi1}.

\begin{figure}
\centering
\includegraphics[width=1\linewidth]{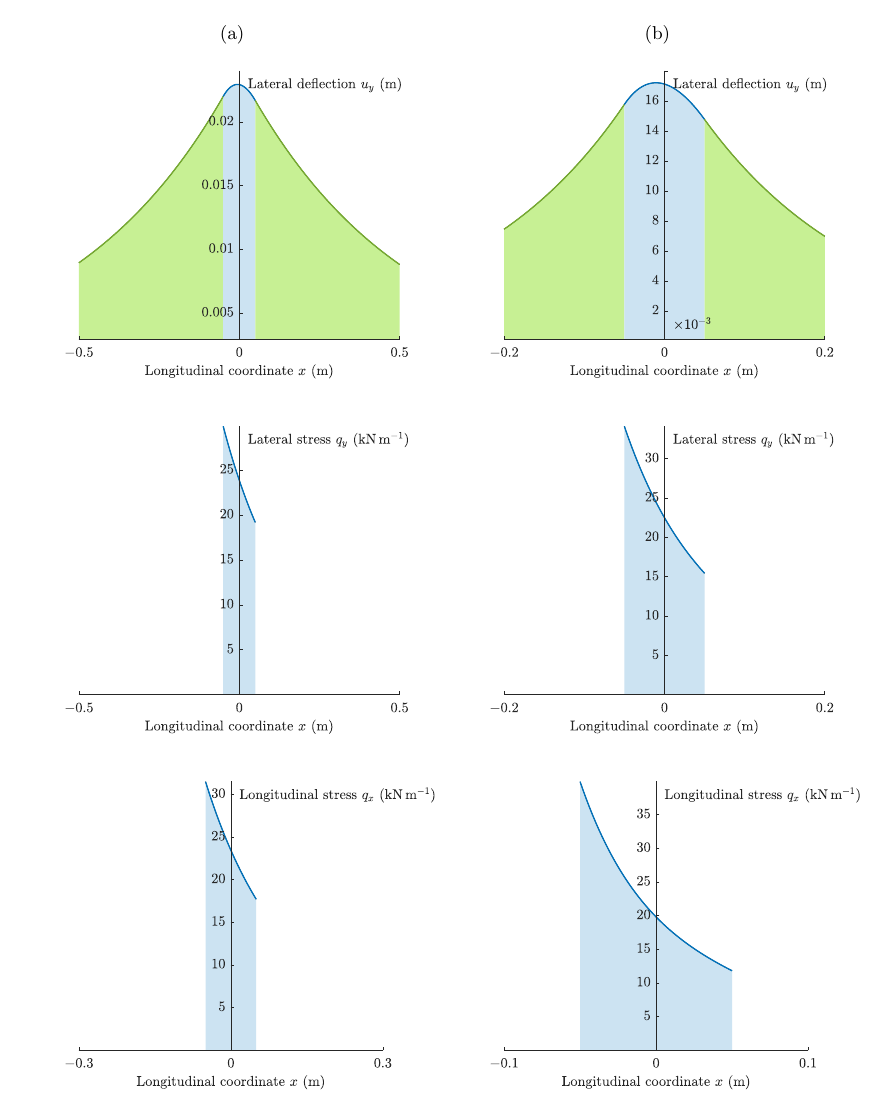} 
\caption{Steady-state lateral deflection and shear stresses. Model parameters as in Table~\ref{tab:parameters}: (a) P1; (b) P2.}
\label{fig:String}
\end{figure}

\begin{figure}
\centering
\includegraphics[width=1\linewidth]{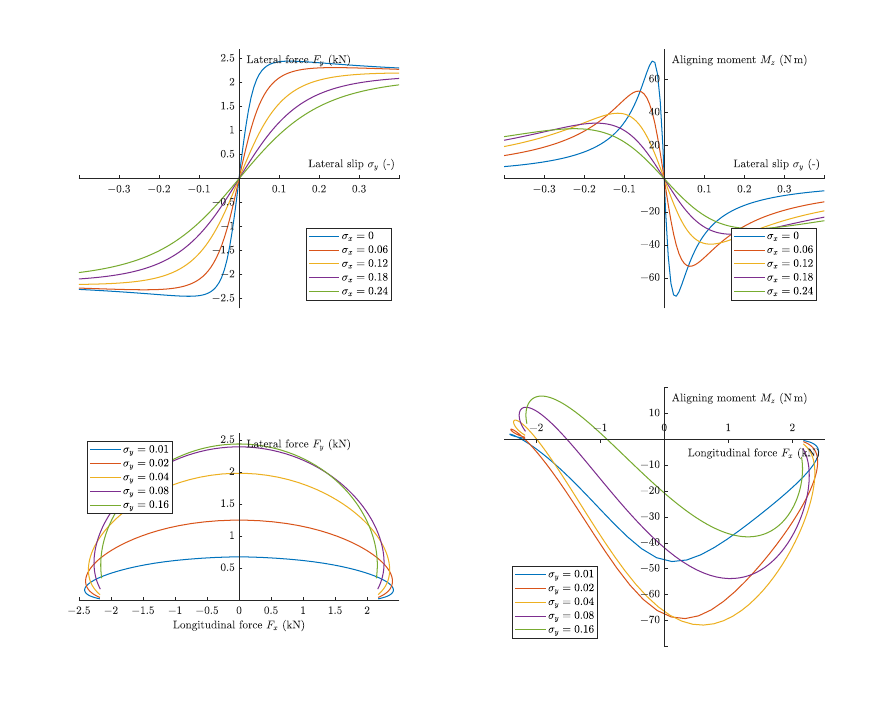} 
\caption{Steady-state tyre characteristics. Model parameters as in Table~\ref{tab:parameters} (P1).}
\label{fig:tirePlots}
\end{figure}

\begin{figure}
\centering
\includegraphics[width=1\linewidth]{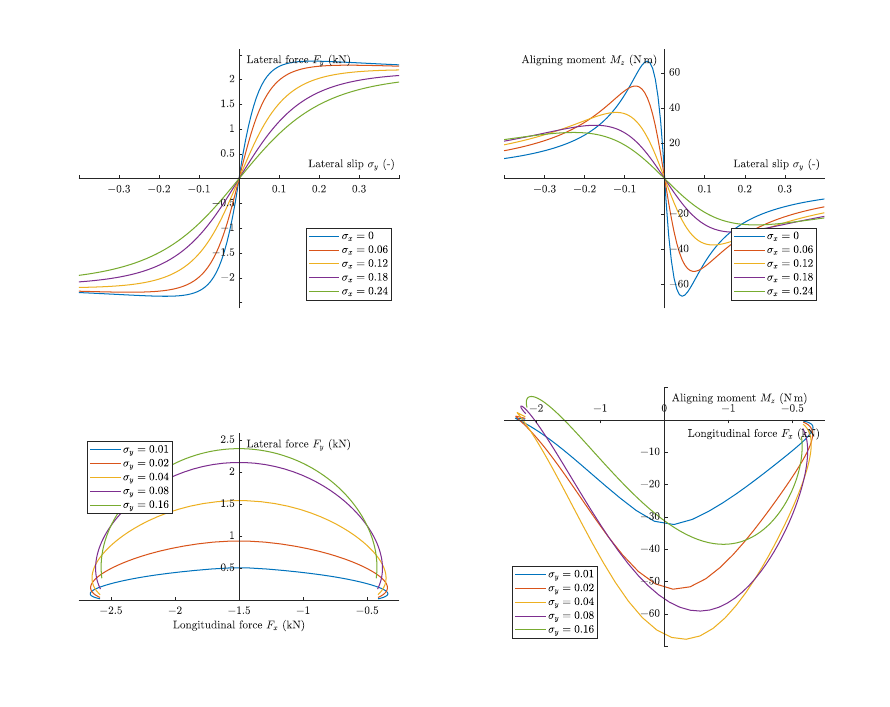} 
\caption{Steady-state tyre characteristics. Model parameters as in Table~\ref{tab:parameters} (P2).}
\label{fig:tirePlots2}
\end{figure}

\begin{table}[]
\centering
\caption{Model parameters for Sect.~\ref{sect:steady}}
{\begin{tabular}{|c|l|c|c|c|}
\hline
Parameter       & Description                    & Unit                                  & Value (P1)        & Value (P2)        \\
\hline
$k_x$           & Longitudinal stiffness         & $\textnormal{N}\,\textnormal{m}^{-2}$ & $2\times 10^5$     & $6\times 10^5$     \\
$k_y$           & Lateral stiffness              & $\textnormal{N}\,\textnormal{m}^{-2}$     & $1\times 10^5$     & $3\times 10^5$     \\
$\lambda_x$     & Longitudinal relaxation length & m                                     & 0.3               & 0.1               \\
$\lambda_y$     & Lateral relaxation length      & m                                     & 0.5               & 0.2               \\
$EA$            & Truss stiffness              & N                                     & $1.8\times 10^4$   & $6\times 10^3$     \\
$S$             & Effective string tension       & N                                     & $2.5\times 10^4$   & $1.2\times 10^4$   \\
$a$             & Contact patch semilength       & m                                     & 0.05               & 0.05               \\
$V\ped{r}$      & Rolling speed                  & $\textnormal{m}\,\textnormal{s}^{-1}$ & 16                & 16                \\
$\mu\ped{s}$    & Static friction coefficient    & -                                     & 1                 & 1                 \\
$\mu\ped{d}$    & Dynamic friction coefficient   & -                                     & 0.7               & 0.7               \\
$v\ped{S}$      & Stribeck velocity              & $\textnormal{m}\,\textnormal{s}^{-1}$ & 3.49              & 3.49              \\
$\delta\ped{S}$ & Stribeck exponent              & -                                     & 0.6               & 0.6               \\
$F_z$           & Vertical force                 & N                                     & 3000              & 3000              \\
$\varepsilon$   & Regularisation parameter       & -                                     & $1\times 10^{-12}$ & $1\times 10^{-12}$ \\
\hline
\end{tabular}}
\label{tab:parameters}
\end{table}
Figure~\ref{fig:tirePlots3} illustrates a comparison between the tyre characteristics predicted by the proposed FrSD model (solid lines) and an isotropic version of Pacejka's Magic Formula (MF) (dashed lines), describing the tangential forces and aligning moment as
\begin{subequations}
\begin{align}
\bm{F}_{\bm{t}}(\bm{\sigma}) & = F_{\bm{t}}(\sigma)\dfrac{\bm{\sigma}}{\sigma}, \\
M_z(\bm{\sigma}) & = -t\ped{p}(\bm{\sigma})F_y(\bm{\sigma})-\dfrac{F_x(\bm{\sigma})F_y(\bm{\sigma})}{K_y}, \quad \bm{\sigma} \in \mathbb{R}^2, 
\end{align}
\end{subequations}
where $\mathbb{R}_{\geq 0} \ni \sigma \triangleq \norm{\bm{\sigma}}_2$ is the total theoretical slip, $K_y \in \mathbb{R}_{>0}$ represent the lateral stiffness of the tyre carcass, modelled as an equivalent lumped elastic element, and the total tangential force $\mathbb{R}_{\geq 0} \ni F_{\bm{t}} \triangleq \norm{\bm{F}_{\bm{t}}}_2$ and pneumatic trail $t\ped{p}(\bm{\sigma}) \in \mathbb{R}$ read \cite{Pacejka2}
\begin{subequations}
\begin{align}
F_{\bm{t}}(\sigma) & = D\sin\Bigl(C\atan\Bigl(B\sigma -E\bigl(B\sigma-\atan(B\sigma)\bigr)\Bigr)\Bigr), \\
t\ped{p}(\bm{\sigma}) & = D_t\cos\Bigl(C_t\atan\Bigl(B_t\sigma -E\bigl(B_t\sigma-\atan(B_t\sigma)\bigr)\Bigr)\Bigr)\cos\sigma_y, \quad \bm{\sigma} \in \mathbb{R}^2, 
\end{align}
\end{subequations}
where the coefficients $B,C,D,E, B_t,C_t,D_t \in \mathbb{R}_{>0}$ and $E, E_t \in \mathbb{R}$ have the usual meaning. The MF model parameters used to produce Fig.~\ref{fig:tirePlots3} were optimised using a genetic program in MATLAB\textsuperscript{\textregistered} with 500 iterations, and are listed in Table~\ref{tab:paramMF}. 

\begin{figure}
\centering
\includegraphics[width=0.8\linewidth]{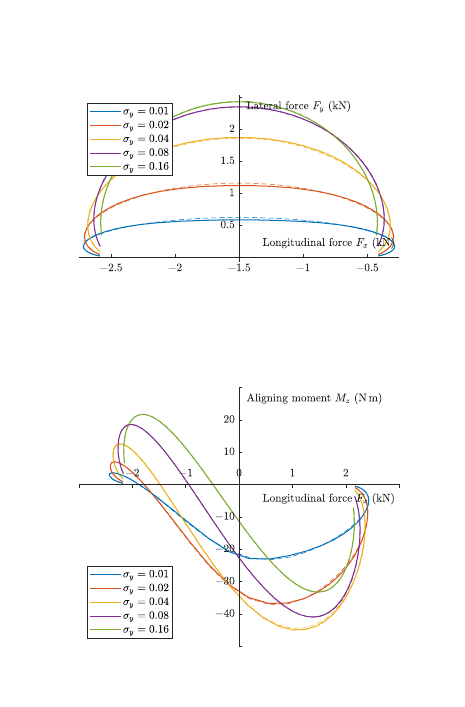} 
\caption{Comparison between the proposed FrSD model and Pacejka's Magic Formula (MF). Model parameters as in Tables~\ref{tab:parameters} (with $k_x=k_y = 1\times 10^5$ $\textnormal{N}\,\textnormal{m}^2$ and $\lambda_x = \lambda_y = 0.5$ m) and~\ref{tab:paramMF}. Line styles: FrSD (solid lines), MF (dashed lines).}
\label{fig:tirePlots3}
\end{figure}

\begin{table}[]
\centering
\caption{MF model parameters}
{\begin{tabular}{|c|l|c|c|}
\hline
Parameter       & Description                    & Unit                                  & Value             \\
\hline
$B$           & Stiffness factor (force)         & -    & 19.37    \\
$C$     & Shape factor (force) & -                                    & 1.36                        \\
$D$             & Peak value (force)       & N                                     & 2440                              \\
$E$    & Curvature factor    & -                                     & 0.27                                 \\
$B_t$    & Stiffness factor (pneumatic trail)   & -                                     & 39.57                             \\
$C_t$      & Shape factor (pneumatic trail)             & $\textnormal{m}\,\textnormal{s}^{-1}$ & 1.00                        \\
$D_t$ & Peak value (pneumatic trail)             & m                                     & 0.04                              \\
$E_t$           & Curvature factor (pneumatic trail)              &-                                     & -0.28                          \\
$K_y$   & Tyre carcass stiffness      & $\textnormal{N}\,\textnormal{m}^{-1}$                                    & $1.07\times 10^{5}$  \\
\hline
\end{tabular}}
\label{tab:paramMF}
\end{table}

\subsubsection{Transient behaviour}\label{sect:trans}
Compared to the model variants introduced in \cite{2DFrBD,2DFrBD3}, which rely on local rheological descriptions of the carcass and tread deflections, the proposed FrSD model can capture relaxation phenomena more accurately owing to the diffusive properties of Eq.~\eqref{es:PDEsos1PDE}, which are inherited from the nonlocal constitutive equation~\eqref{eq:ODE_SC}. This feature permits accounting for nonlinear unsteady responses excited by large slip inputs, which generate short-lived transients. In this context, Fig.~\ref{fig:trans1} depicts the evolution of the lateral tyre force and aligning moment following step changes in slip of varying magnitude. Inspection of Fig.~\ref{fig:trans1} shows that increasing slip levels lead to faster relaxation dynamics. This behavior persists even for slip inputs exceeding those associated with the peak values of the tyre characteristics, as evidenced, for instance, by the responses of both $F_y$ and $M_z$ for $\sigma_y = 0.2$.

A further distinctive capability of the FrSD formulation is its tendency to mitigate or attenuate oscillations induced by time-varying slip inputs, which is in agreement with the dynamical behaviour predicted by lumped ODE descriptions such as the single-contact point and two-regime tyre models \cite{Pacejka2,Guiggiani,LibroMio}. This effect is illustrated in Fig.~\ref{fig:trans2}, where the tyre is subjected to a sinusoidal input of the form $\sigma_y(s) = \bar{\sigma}_y[1 + 0.5\sin(\omega s)]$, with $\omega = 5$ $\textnormal{m}^{-1}$. Particularly at large values of $\bar{\sigma}_y$, the system exhibits a pronounced low-pass filtering behaviour. From a mathematical standpoint, this property originates again from the nonlinear diffusive term in Eq.~\eqref{es:PDEsos1PDE}. Near saturation, fluctuations around the steady-state lateral force are further reduced due to the high friction utilisation. For a fixed value of $\bar{\sigma}_y$, higher-frequency oscillations in the slip input are increasingly attenuated. This trend is evident in Fig.~\ref{fig:trans3}, which shows the transient responses of the lateral force and vertical moment for $\bar{\sigma}_y = 0.08$ and $\omega = 5$, $10$, and $20~\textnormal{m}^{-1}$. The phenomena visualised in Figs.~\ref{fig:trans2} and~\ref{fig:trans3} should not be interpreted as evidence that the proposed models can genuinely suppress every type of oscillation. Higher-order formulations including inertia may predict genuine oscillatory contact dynamics, and such effects can also be observed experimentally. 
\begin{figure}
\centering
\includegraphics[width=0.9\linewidth]{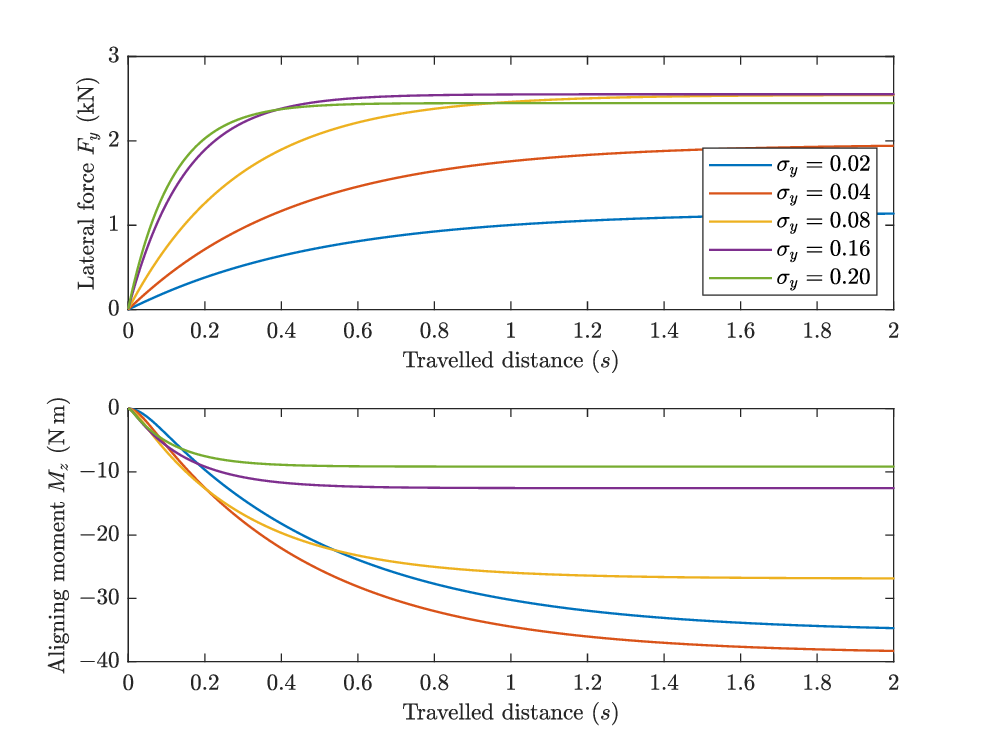} 
\caption{Transient response to different step slip inputs. Model parameters as in Table~\ref{tab:parameters}.}
\label{fig:trans1}
\end{figure} 

\begin{figure}
\centering
\includegraphics[width=0.9\linewidth]{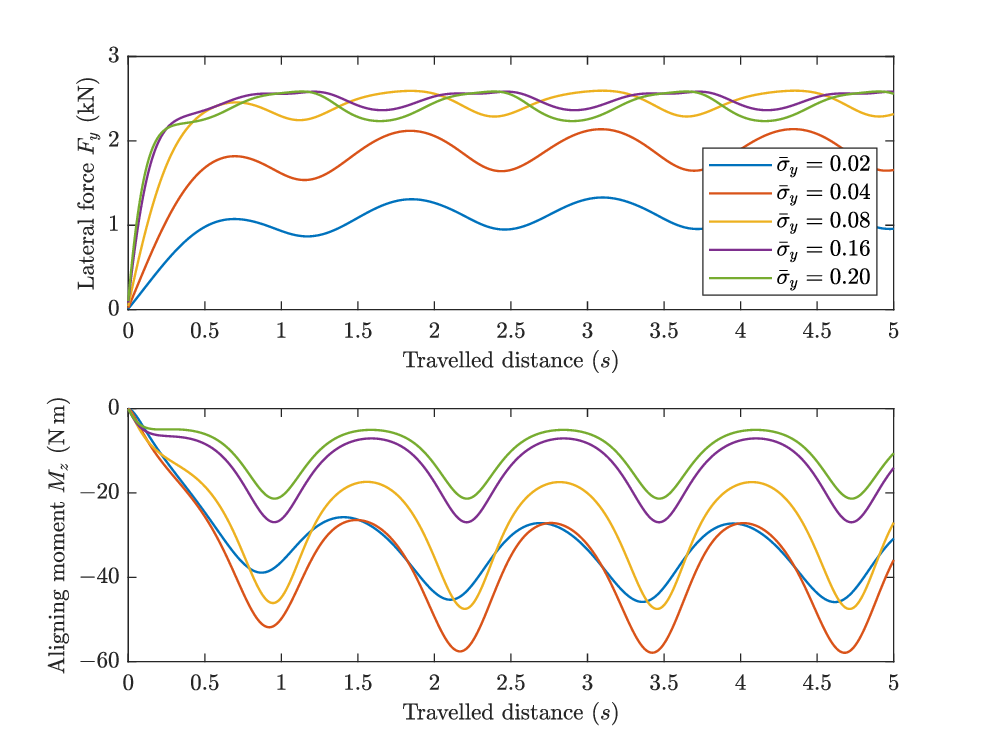} 
\caption{Transient lateral response to a sinusoidal slip input $\sigma_y(s) = \bar{\sigma}_y[1 + 0.5\sin(\omega s)]$ for different values of $\bar{\sigma}_y$ and $\omega = 5$ $\textnormal{m}^{-1}$. Model parameters as in Table~\ref{tab:parameters}.}
\label{fig:trans2}
\end{figure} 

\begin{figure}
\centering
\includegraphics[width=0.9\linewidth]{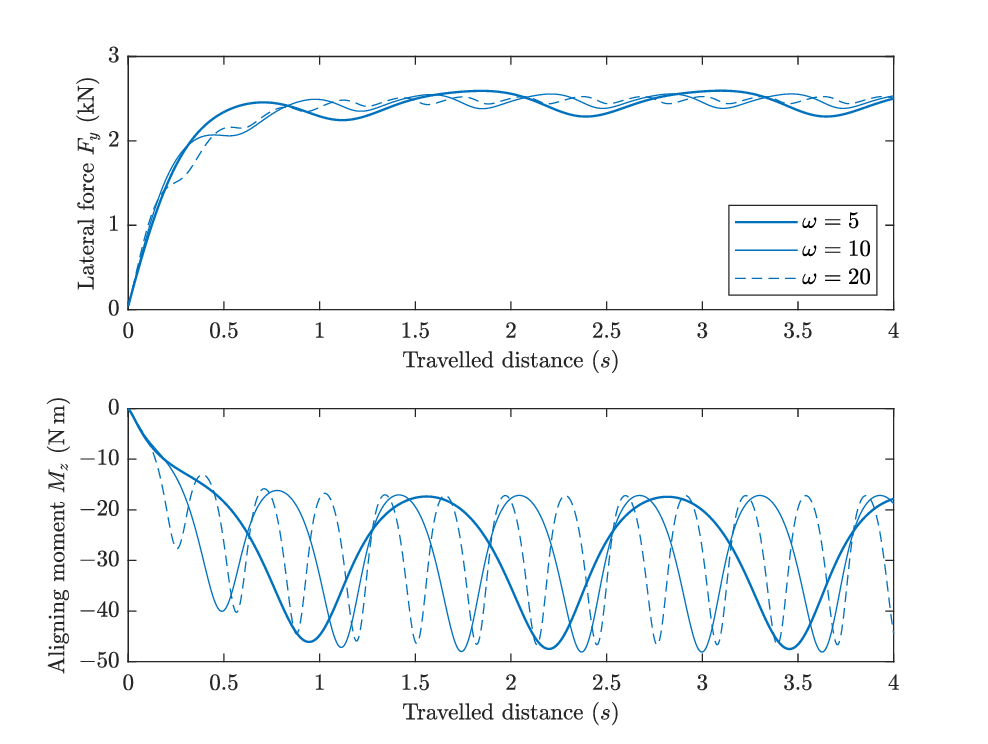} 
\caption{Transient lateral response to a sinusoidal slip input $\sigma_y(s) = \bar{\sigma}_y[1 + 0.5\sin(\omega s)]$ for $\bar{\sigma}_y = 0.08$ and different values of $\omega$. Model parameters as in Table~\ref{tab:parameters}.}
\label{fig:trans3}
\end{figure}

\subsection{Experimental validation}\label{sect:exp}
The FrSD model was validated against data collected from a test rig \cite{Besselink,thesis1,thesis2}. The next Sects.~\ref{sect:carcassOpt} and~\ref{sect:relaxExp} present experimental results concerning the lateral tyre deflection and transient response to step slip inputs, respectively.

\subsubsection{Steady-state tyre carcass deflection}\label{sect:carcassOpt}
The carcass and sidewall elements may undergo relatively large deformations when a tyre is subjected to lateral slip, i.e., to a sideslip angle ($\tan \alpha = \sigma_y$ in the absence of longitudinal slip). For the purpose of model validation, the tyre deflection was measured with high accuracy using optical devices, specifically the setup described in \cite{Besselink,thesis1,thesis2}. Figure~\ref{fig:ssDeflExt} compares the experimentally measured lateral deformations to the steady-state response predicted by the FrSD model (dashed lines), for several values of the sideslip angle.
In the vicinity of the contact patch ($\abs{\theta-180} < 10^\circ$), the FrSD model qualitatively reproduces both the magnitude and the profile of the measured deformation. For larger values of $\abs{\theta-180}$, the approximation introduced by assuming exponentially decaying functions for $\abs{x} > a$ results in noticeable deviations. In particular, the fitted relaxation lengths are larger than those suggested directly by the measured deformation profiles; this discrepancy may reflect some unmodelled mechanical effect. Nevertheless, the agreement outside the contact patch is of limited relevance for the evaluation of tyre forces and moments, since, according to Eq.~\eqref{eq:FMineg}, these quantities can be determined solely from the deflection profile within the contact patch and the boundary conditions at the leading and trailing edges.

The deformation trends illustrated in Fig.~\ref{fig:ssDeflExt} were obtained by calibrating the FrSD model parameters using a genetic algorithm implemented in MATLAB\textsuperscript{\textregistered}. The identified parameters, listed in Table~\ref{tab:parameters2}, were optimised to reproduce the experimental deformation profiles over the angular range $\theta \in [135^\circ, 225^\circ]$, in accordance with the observations reported above.
\begin{figure}
\centering
\includegraphics[width=0.8\linewidth]{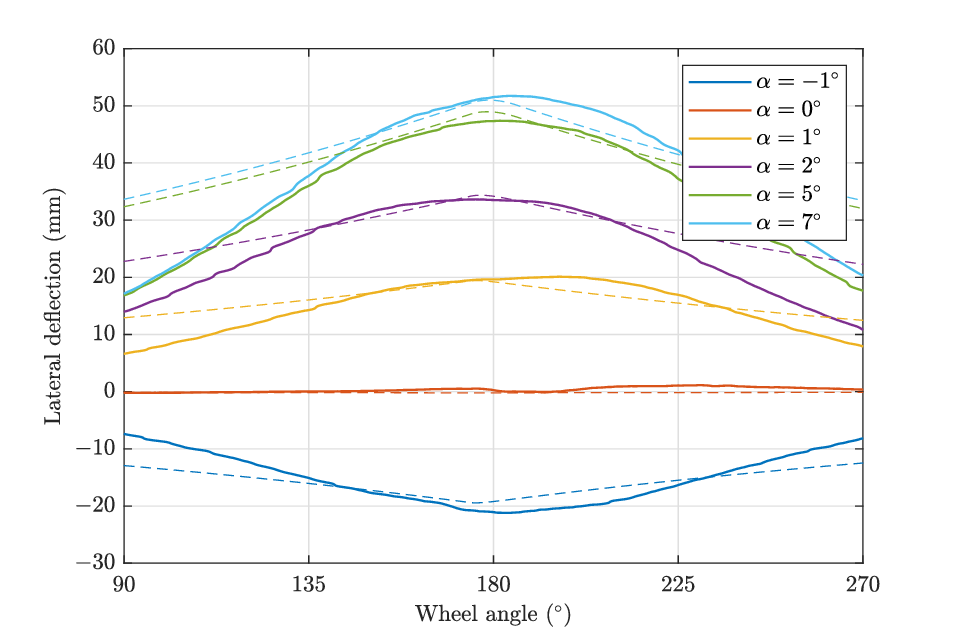} 
\caption{Steady-state lateral carcass deflection for different tyre sideslip angles. Line styles: Experimental (thick solid lines), FrSD model (dashed lines). Model parameters as in Table~\ref{tab:parameters2}.}
\label{fig:ssDeflExt}
\end{figure}

\begin{table}[]
\centering
\caption{Model parameters for Sect.~\ref{sect:carcassOpt}}
{\begin{tabular}{|c|l|c|c|}
\hline
Parameter       & Description                    & Unit                                  & Value              \\
\hline
$k_y$           & Lateral stiffness              & $\textnormal{N}\,\textnormal{m}^{-2}$    & $3.00\times 10^4$      \\
$\lambda_y$     & Lateral relaxation length      & m                                     & 1.089                             \\
$a$             & Contact patch semilength       & m                                     & 0.03                              \\
$\mu\ped{s}$    & Static friction coefficient    & -                                     & 1.03                                 \\
$\mu\ped{d}$    & Dynamic friction coefficient   & -                                     & 0.72                             \\
$v\ped{S}$      & Stribeck velocity              & $\textnormal{m}\,\textnormal{s}^{-1}$ & 10                         \\
$\delta\ped{S}$ & Stribeck exponent              & -                                     & 2                              \\
$F_z$           & Vertical force                 & N                                     & 3700                          \\
$\varepsilon$   & Regularisation parameter       & -                                     & $1\times 10^{-12}$  \\
\hline
\end{tabular}}
\label{tab:parameters2}
\end{table}

\subsubsection{Relaxation behaviour}\label{sect:relaxExp}
The relaxation behaviour of the tyre was investigated experimentally by analysing the transient response to practical longitudinal step slip inputs $\kappa_x$ ($\sigma_x = \frac{\kappa_x}{1+\kappa_x}$ in the absence of lateral slip). Figure~\ref{fig:expRel} compares the experimental data with the predictions of the FrSD model, whose parameters (listed in Table~\ref{tab:parameters3}) were optimised using a genetic algorithm implemented in MATLAB\textsuperscript{\textregistered}, with a population size of 100 individuals and 300 generations. 

Inspection of Fig.~\ref{fig:expRel} indicates that the FrSD model captures the nonlinear relaxation effects with satisfactory accuracy across the entire range of slip inputs, correctly predicting convergence to the true steady-state longitudinal force. This also confirms the model's ability to reproduce the stationary tyre characteristic over a wide range of slip values. Some discrepancies are nevertheless observed concerning the transient behaviour at larger slip values, where the optimised model exhibits slower transient dynamics. It is also worth noting that the response is not symmetric with respect to the slip input, with the FrSD description slightly overestimating the magnitude of the steady-state longitudinal force for $\kappa_x = -0.1$ and $\kappa_x = -0.05$.
\begin{figure}
\centering
\includegraphics[width=0.7\linewidth]{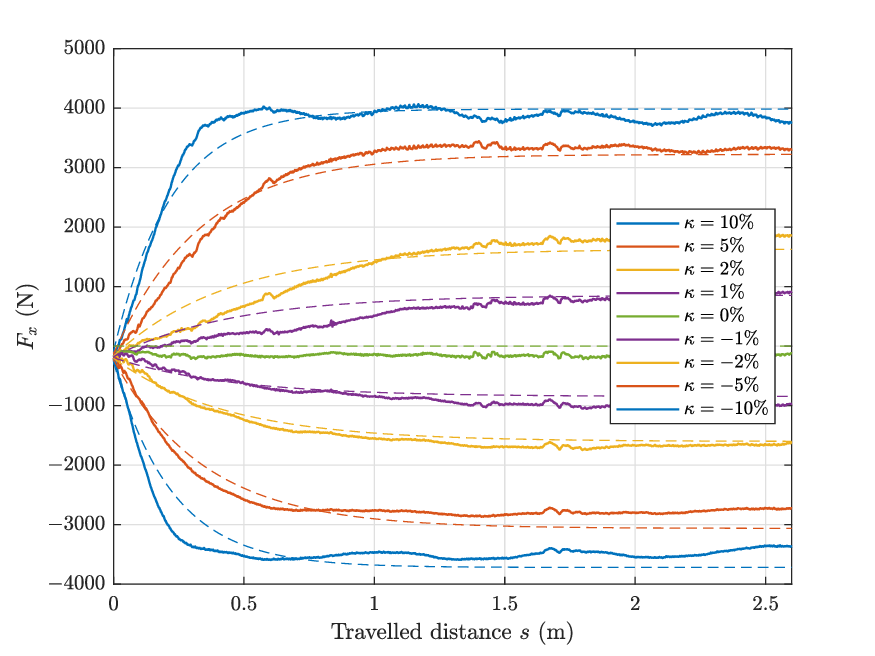} 
\caption{Longitudinal relaxation behaviour of the tyre subjected to different practical step slip inputs. Line styles: Experimental (thick solid lines), FrSD model (dashed lines). Model parameters as in Table~\ref{tab:parameters3}.}
\label{fig:expRel}
\end{figure} 

\begin{table}[]
\centering
\caption{Model parameters for Sect.~\ref{sect:relaxExp}}
{\begin{tabular}{|c|l|c|c|}
\hline
Parameter       & Description                    & Unit                                  & Value             \\
\hline
$k_x$           & Longitudinal stiffness         & $\textnormal{N}\,\textnormal{m}^{-2}$    & $1.68\times 10^5$     \\
$\lambda_x$     & Longitudinal relaxation length & m                                     & 0.51                        \\
$a$             & Contact patch semilength       & m                                     & 0.05                              \\
$\mu\ped{s}$    & Static friction coefficient    & -                                     & 0.83                                 \\
$\mu\ped{d}$    & Dynamic friction coefficient   & -                                     & 0.58                             \\
$v\ped{S}$      & Stribeck velocity              & $\textnormal{m}\,\textnormal{s}^{-1}$ & 1.07                        \\
$\delta\ped{S}$ & Stribeck exponent              & -                                     & 2                              \\
$F_z$           & Vertical force                 & N                                     & 3700                          \\
$\varepsilon$   & Regularisation parameter       & -                                     & $1\times 10^{-12}$  \\
\hline
\end{tabular}}
\label{tab:parameters3}
\end{table}

\section{Conclusions}\label{sect:conclusions}
This paper presented a new string-like tyre model with distributed friction dynamics, renamed FrSD (\emph{friction with string dynamics}). By combining a nonlocal constitutive representation of the tyre structural deformation with a dynamic friction law of FrBD type, the proposed formulation departs from classical transport-dominated rolling contact models and leads to a system of semilinear parabolic partial differential equations.

A central outcome of this work is the recognition that the inclusion of higher-order spatial derivatives in the constitutive tyre equations fundamentally alters the mathematical character of the rolling contact process. In contrast to traditional hyperbolic formulations, the parabolic structure of the FrSD model naturally introduces diffusive effects that govern the evolution of the treadband and carcass deflections within the contact patch. These mechanisms provide a physically consistent explanation for relaxation phenomena that depend on slip and spin magnitude, as well as for the observed attenuation of high-frequency slip excitations.
The proposed formulation retains the key advantages of the string and brush tyre models -- including physical interpretability and a clear link between local deformation and global force generation -- whilst addressing several of their well-known limitations. In particular, the ability to impose boundary conditions at both ends of the contact patch yields smooth deformation profiles, avoiding artificial discontinuities at the leading or trailing edge. Moreover, the regularised FrBD friction law ensures well-posedness and numerical robustness across the full range of operating conditions.

Numerical simulations demonstrated that the FrSD model reproduces classical steady-state tyre characteristics under combined slip conditions, including force saturation and aligning moment behaviour. At the same time, it captures distinctive transient features such as slip-dependent relaxation speeds and a pronounced low-pass filtering of oscillatory inputs. These properties are particularly relevant for vehicle dynamics simulations involving aggressive manoeuvres, stability control, and high-bandwidth actuation.

From a vehicle systems perspective, the proposed formulation provides a promising basis for the development of tyre models that are both physically grounded and suitable for integration into control-oriented and multibody simulation frameworks. The parabolic structure facilitates theoretical analysis, including stability and dissipativity studies, and may prove advantageous in the context of reduced-order modelling and real-time applications.
Future work will focus on extensions to account for additional physical effects such as camber, pressure variations, and viscoelastic tread behaviour. The series connection between string-like and bristle elements should also be implemented to improve the relaxation behaviour of the tyre model. Further investigation of the implications of the parabolic structure for vehicle-level stability and control design also represents a promising direction for ongoing research.

\section*{Funding declaration}
This research was financially supported by the project FASTEST (Reg. no. 2023-06511), funded by the Swedish Research Council. 

%Luigi Romano also thanks Massimo Guiggiani for valuable discussions on the topic, as well as his inspiring work on the fundamentals of Vehicle Dynamics.
\section*{Acknowledgements}

The author thanks Igo Besselink at TU/e for providing the experimental data used in Sect.~\ref{sect:exp}.

\section*{Compliance with Ethical Standards}

The authors declare that they have no conflict of interest.

\section*{Author Contribution declaration}
L.R. is the sole author of and contributor to the manuscript.

\appendix

\end{document}